\documentclass[a4paper, thm-restate, 11pt]{article}
\usepackage[bottom,para,symbol]{footmisc}
\usepackage{authblk}
\usepackage{todonotes}
\usepackage{multirow}
\usepackage{etoolbox}
\usepackage{amsmath, amssymb, amsthm,complexity}
\usepackage{graphicx}
\usepackage{xcolor}
\usepackage{tcolorbox}
\tcbuselibrary{skins} 
\usepackage{algorithm}
\usepackage[noend]{algpseudocode}
\usepackage{thm-restate}
\newtheorem{remark}{Remark}
\newtheorem{mytheorem}{Theorem}
\newtheorem{mycorollary}{Corollary}
\newtheorem{mylemma}{Lemma}[section]

\newtheorem{myobservation}{Observation}[section]

\usepackage{mathtools}
\usepackage{bold-extra}

\usepackage{fullpage}

\newcommand{\no}{\textsc{No}}

\newcommand{\LR}[1]{\left\{#1\right\}}

\newcommand{\blue}[1]{{\color{blue} #1}}

\algnewcommand{\LeftComment}[1]{\Statex \(\triangleright\) #1}
\newcommand{\blueLcomment}[1]{\LeftComment{ \blue{\small\texttt{#1}}}}
\algrenewcommand\algorithmicrequire{\textbf{Input:}}
\algrenewcommand\algorithmicensure{\textbf{Output:}}

\newcommand{\parprobNoPar}[4]{
\vspace{2mm}
\noindent\fbox{
\begin{minipage}{0.96\textwidth}
\begin{tabular*}{\textwidth}{@{\extracolsep{\fill}}lr} \textsc{#1}  & {\bf{ }} #2
\\ \end{tabular*}
{\bf{Input:}} #3  \\
{\bf{Question:}} #4
\end{minipage}
}
\vspace{2mm}
}

\title{Arborescences and Shortest Path Trees when Colors Matter\thanks{This work is supported by SERB CRG grant number CRG/2022/007751.\hfill}}

\author{P.~S.~Ardra\footnote{111914001@smail.iitpkd.ac.in}~}

\author{Jasine~Babu\footnote{jasine@iitpkd.ac.in}~}

\author{Kritika~Kashyap\footnote{112202003@smail.iitpkd.ac.in}~}

\author{R.~Krithika\footnote{krithika@iitpkd.ac.in}~}

\author{Sreejith~K.~Pallathumadam\footnote{111704002@smail.iitpkd.ac.in}~}

\author{Deepak~Rajendraprasad\footnote{deepak@iitpkd.ac.in}}

\affil{Indian Institute of Technology Palakkad, Palakkad, India.}

\theoremstyle{plain}
\newtheorem{theorem}{Theorem}

\newtheorem{definition}[theorem]{Definition}

\newtheorem{reduction rule}{Reduction Rule}[section]

\usepackage{hyperref}
\date{}

\begin{document}
\maketitle
\begin{abstract}
We are given an edge-colored (directed or undirected) graph and our objective is to find a specific type of subgraph, like a spanning tree, an arborescence, a single-source shortest path tree, a perfect matching etc., with constraints on the number of edges of each color. Some of these problems, like color-constrained spanning tree, have elegant solutions and some of them, like color-constrained perfect matching, are longstanding open questions. In this work, we study color-constrained arborescences and shortest path trees. Computing a color-constrained shortest path tree on weighted digraphs turns out to be \NP-hard in general but polynomial-time solvable when all cycles have positive weight. This polynomial-time solvability is due to the fact that the solution space is essentially the set of all color-constrained arborescences of a directed acyclic subgraph of the original graph. While finding color-constrained arborescence of digraphs is \NP-hard in general, we give an efficient algorithm when the input digraph is acyclic. Consequently, a color-constrained shortest path tree on weighted digraphs having only positive weight cycles can be efficiently computed. En route, we sight nice connections to colored matroids and color-constrained bases. Our algorithm also generalizes to the problem of finding a color-constrained shortest path tree with the minimum total weight. Both our algorithms use a single source shortest path algorithm and a (minimum cost) maximum flow algorithm as subroutines. By using the recent algorithm by van den Brand et al. (FOCS 2023) for these subroutines, our algorithms achieve near-linear running time when the edge weights are integral and polynomially-bounded in the size of the graph.
\end{abstract}
\section{Introduction}
\label{sec:intro}
Edge-colored graphs model scenarios where different kinds of binary relationships exist among the underlying set of entities. 
We call a graph whose edges are colored with elements of $[q]:=\{1,2,\ldots,q\}$ as a {\em $q$-colored graph} and in this work we focus on finding color-constrained subgraphs in $q$-colored graphs. For example, consider the problem where given an undirected $q$-colored graph $G$ and a vector $\alpha=(\alpha_1,\ldots,\alpha_q) \in \mathbb{N}^q$, the objective is to find a spanning tree that has exactly $\alpha_i$ edges of color $i$ for each $i \in[q]$. Papadimitriou and Yannakakis \cite{PapadimitriouY82} gave an elegant solution to this problem when $q=2$ by efficiently constructing a sequence of spanning trees $T_{min}=T_1,T_2,\ldots,T_{max}$. Here, $T_{min}$ and $T_{max}$ denote spanning trees with minimum and maximum number of edges of color $1$, and $T_{i+1}$ is obtained from $T_i$ by replacing an edge colored 2 with an edge colored 1. One of the spanning trees in this sequence is a solution (if one exists). It turns out that this ``exchange sequence'' exists in the more general setting of matroids as observed by Gabow and Tarjan~\cite{GabowT84} who used it to give a polynomial-time algorithm that computes a minimum weight spanning tree that has $\alpha_1$ edges of color~$1$ and $\alpha_2$ edges of color~$2$. An alternative algorithm for the same problem was given by Gusfield~\cite{Gusfield84}. Later, Rendl and Leclerc gave a polynomial-time algorithm for arbitrary number of colors~\cite{RendlL89}. Since the set of all forests in $G$ forms a graphic matroid on $E(G)$, the spanning tree that we seek is a largest ($|V(G)|-1$ sized) common independent set of this graphic matroid and the partition matroid encoding the color constraints.  Thus, the problem of finding a spanning tree that has exactly $\alpha_i$ edges of color $i$ for each $i \in[q]$ may be viewed as a \textsc{2-Matroid Intersection} problem of determining a largest common independent set of two matroids which is well-known to be polynomial-time solvable~\cite{Schrijver03}. 

Finding a color-constrained spanning out-tree rooted at $s$ (\textit{$s$-arborescence} that has exactly $\alpha_i$ edges of color $i$) of a $q$-colored digraph $G$ has an interesting difference in complexity when compared to its undirected analogue. Barahona and Pulleyblank~\cite{BarahonaP87} gave a polynomial-time algorithm for this problem when $q=2$. The set of all out-forests of $G$ (with no edges entering $s$) do not necessarily form a matroid with $E(G)$ as the ground set. However, an out-forest of $G$ (with no edges entering $s$) is a common independent set of two matroids; the first one being a graphic matroid (obtained by ignoring the directions in $G$) and the second one being a partition matroid of all subgraphs of $G$ in which no edges enter $s$ and at most one edge enters every other vertex. Then the desired $s$-arborescence is a solution to \textsc{3-Matroid Intersection} (where the objective is to determine a largest common independent set of three matroids). \textsc{3-Matroid Intersection} is \NP-complete~\cite{Schrijver03} in general and we observe that this hardness remains even in the special case of finding a color-constrained $s$-arborescence (Theorem~\ref{thm:col-arb-hard}). The complexity of  finding color-constrained $s$-arborescences in $q$-colored digraphs for a fixed value of $q>2$ remains open.

This prompts one to study the complexity of finding color-constrained shortest path trees which are a very commonly occurring type of arborescences. Interestingly, this problem turns out to be efficiently solvable even when there are upper and lower bound constraints on the number of edges of each color, as long as all cycles in the digraph have positive weight. Before we define this problem formally, we introduce some related terminology. Given a vector $\alpha=([\alpha_1,\alpha'_1],\ldots,[\alpha_q,\alpha'_q]) \in (\mathbb{N} \times \mathbb{N})^q$ of intervals (referred to as a {\em color-constraint vector}), a $q$-colored graph is said to be an {\em $\alpha$-colored graph} if for each $i \in[q]$ it has at least $\alpha_i$ but at most $\alpha'_i$ edges of color $i$. We call a graph associated with an edge weighting $w:E(G) \rightarrow \mathbb{R}$ a $w$-weighted graph. A \emph{shortest path tree} rooted at $s$ (referred to as $s$-shortest path tree and abbreviated as $s$-SPT) of $G$ is an out-tree such that for every vertex $v$ reachable from $s$, the $s$-$v$-path in the tree is a shortest (minimum weight) $s$-$v$-path in $G$. Recall that vertices are not allowed to repeat in a path. For simplicity, we assume that all vertices of $G$ are reachable from the source vertex $s$ and therefore $|E(G)| \geq |V(G)|-1$. Further, we assume that $q \leq |E(G)|$.

\parprobNoPar{\textsc{Color-Constrained Shortest Path Tree} (\textsc{CC-SPT})}{}{A $q$-colored $w$-weighted digraph $G$ (with possibly parallel edges), a vertex $s\in V(G)$ designated as source and a color-constraint vector $\alpha=([\alpha_1,\alpha'_1],\ldots,[\alpha_q,\alpha'_q]) \in (\mathbb{N} \times \mathbb{N})^q$.}{Does $G$ have an $\alpha$-colored $s$-shortest path tree?}

 In graphs with negative-weight cycles, subpaths of a shortest path are not necessarily shortest paths. Consequently, the set of shortest paths from a single source may not even define a tree. While shortest path trees exist and can be computed efficiently for graphs without negative-weight cycles (c.f. using Bellman-Ford algorithm), we show that \textsc{CC-SPT} is \NP-hard for this class (Theorem~\ref{thm:spt-q-col-NP-h}) by a reduction from \textsc{Hamiltonian Path}. These make one curious about the complexity of \textsc{CC-SPT} when all cycles have positive weight. Clearly, not every $s$-arborescence of $G$ is an $s$-SPT of $G$. Nevertheless, there is a subgraph $G_s$ of $G$ that contains every $s$-SPT of $G$ with the additional property that every $s$-arborescence of $G_s$ is an $s$-SPT of $G$. We may choose $G_s$ to be the subgraph of $G$ where for any pair $u,v$ of vertices, $(u,v)$ is an edge of $G_s$ if and only if $(u,v)$ is the last edge on some shortest $s$-$v$-path in $G$. It follows that in order to find an $\alpha$-colored $s$-SPT of $G$, it suffices to find an $\alpha$-colored $s$-arborescence of $G_s$. This leads us to the \textsc{Color-Constrained Arborescence} problem defined below. Observe that $G_s$ can be computed by an algorithm that computes an $s$-SPT of $G$ ignoring colors. Further, if all cycles in $G$ are positive, then $G_s$ is a directed acyclic graph (DAG) (Observation~\ref{obs:SPS-DAG}) and therefore we only need to find an $\alpha$-colored $s$-arborescence in a DAG (which is $G_s$) in order to find the desired $s$-SPT of $G$. 

\parprobNoPar{\textsc{Color-Constrained Arborescence} (\textsc{CC-ARB})}{}{A $q$-colored digraph $G_s$ (with possibly parallel edges) where a vertex $s\in V(G_s)$ is designated as the root and a color-constraint vector $\alpha=([\alpha_1,\alpha'_1],\ldots,[\alpha_q,\alpha'_q]) \in (\mathbb{N} \times \mathbb{N})^q$.}{Does $G_s$ have an $\alpha$-colored $s$-arborescence?}

We remark that the graph $G_s$ in the definition of \textsc{CC-ARB} is an arbitrary digraph (and not one necessarily arising from the graph of a \textsc{CC-SPT} instance). Once again we assume that all vertices of $G_s$ (with $n$ vertices and $m$ edges) are reachable from $s$, $m \geq n-1$ and $q \leq m$. As mentioned before, \textsc{CC-ARB} is \NP-complete (Theorem~\ref{thm:col-arb-hard}) even if $\alpha_i=0$ for each $i$. However, we show that \textsc{CC-ARB} on a $q$-colored DAG can be efficiently solved by modeling the problem as a flow problem on an auxiliary network. A {\em flow network} is a digraph $H$ with two distinguished vertices, source $s^*$ (with no incoming edges) and sink $t^*$ (with no outgoing edges), and two non-negative functions $c: V(H) \times V(H) \rightarrow \mathbb{R}^{\geq 0}$ ({\em capacity function}) and $d: V(H) \times V(H) \rightarrow \mathbb{R}^{\geq 0}$ ({\em demand function}) such that $c(x,y)>0$ or $d(x,y)>0$ only when $(x,y) \in E(H)$. Optionally, a flow network $H$ may also be associated with a cost function $\widehat{c}:E(H) \rightarrow \mathbb{R}^{\geq 0}$. A {\em flow} in $H$ is a function $f: V(H) \times V(H) \rightarrow \mathbb{R}^{\geq 0}$ such that for each $(x,y) \in V(H) \times V(H)$, $0 \leq f(x,y) \leq c(x, y)$ ({\em non-negativity and capacity constraints}) and for each $x \in V(H)\setminus \{s^*, t^*\}$, $\sum_{y \in V(H)} f(x,y)=\sum_{y \in V(H)} f(y,x)$ ({\em flow conservation constraints}). The value of a flow $f$ is defined as $|f|=\sum_{v \in V(H)} f(s^*,v)$. A \textit{maximum flow} is one with maximum value and a \textit{minimum cost maximum flow} is a maximum flow that minimizes $\sum_{(u,v) \in E(H)} \widehat{c}(u,v) f(u,v)$. A \textit{feasible flow} is a flow $f$ such that for each $(x,y) \in V(H) \times V(H)$, $f(x,y) \geq d(x, y)$ ({\em demand  constraints}). Now, we are ready to define the problem that is of interest to us.

\parprobNoPar{\textsc{Feasible Flow}}{}{A flow network $(H,s^*,t^*,c,d)$ and a real number $k$.}{Does $H$ have a feasible flow of value $k$?}

 Our algorithm for \textsc{CC-ARB} on a DAG is based on computing a feasible flow in an auxiliary flow network which in turn translates into a maximum flow computation in another network. Given a DAG $G_s$ on $n$ vertices and $m$ edges, we first construct in linear time a flow network $H$ (on $\widehat{n}=\mathcal{O}(n+q)$ vertices which is $\mathcal{O}(m)$ and $\widehat{m}=\mathcal{O}(\min(m,nq))$ edges which is also $\mathcal{O}(m)$) with demand constraints such that $G_s$ has an $\alpha$-colored $s$-arborescence if and only if $H$ has a feasible flow of value $n-1$. Next, we construct in linear (in $\widehat{n}$ and $\widehat{m}$) time another flow network $H'$ with no demand constraints such that $H$ has a feasible flow of value $n-1$ if and only if $H'$ has a flow of value $n-1+\sum_{i=1}^q \alpha_i$. Finally, using a simple transformation we obtain the required $s$-arborescence of $G_s$ from the flow computed in $H'$. The overall running time of the algorithm for \textsc{CC-ARB} is $\mathcal{O}(m+T_{MF}(\widehat{n},\widehat{m}))$ (Theorem~\ref{thm:arb-q-col-flow}) where $T_{MF}(n,m)$ denotes the running time of an algorithm that computes an integral maximum flow in a network with $n$ vertices and $m$ edges. This results in an $\mathcal{O}(T_{SSP}(n_G,m_G)+T_{MF}(n_G+q,\min(m_G,n_Gq)))$-time algorithm for \textsc{CC-SPT} on digraphs $G$ on $n_G$ vertices, $m_G$ edges and only positive cycles (Theorem~\ref{thm:spt-q-col-flow}) where $T_{SSP}(n,m)$ denotes the running time of an algorithm that computes an $s$-SPT of a digraph on $n$ vertices and $m$ edges.

 It is well-known that the single source shortest path problem reduces to a minimum cost maximum flow computation\footnote{To find distances from $s$ to every other vertex in $G$, it suffices to solve for a minimum cost flow of value $|V(G)|-1$ in the network $G'$ obtained from $G$ by adding a new sink vertex $t'$ and a unit-capacity zero-cost edge from every vertex in $V(G)\setminus \{s\}$ to $t'$. The edges of $G$ that are in $G'$ have capacities $|V(G)|-1$ and costs equal to their weights.}. By using the algorithm by van den Brand et al.~\cite{Brand0PKLGSS23} for \textsc{Minimum Cost Maximum Flow} as a subroutine twice (once for computing single source shortest path distances\footnote{One may use recent near-linear time single source shortest path algorithms~\cite{BernsteinNW22,BringmannCF23}. However, these are randomized (Las Vegas) algorithms.} in $G$ and once for solving \textsc{CC-ARB} in $G_s$), it follows that our algorithm for \textsc{CC-SPT} has near-linear (in $m_G$) running time when the edge weights are integral and polynomially-bounded in the size of the graph provided the graph has only positive cycles. Note that if $w$ is non-negative, then simpler algorithms (such as Dijkstra's algorithm) may be used to compute single source shortest path distances without any increase in the overall running time. We also show that \textsc{CC-SPT} for $q=2$ can be solved as fast as finding an $s$-SPT in uncolored graphs and does not require a maximum flow  computation (Sec.~\ref{sec:rb}). 
 
 Finding an $s$-SPT of minimum total weight is polynomial-time solvable on non-negative weighted graphs~\cite{KhullerRY95}. We observe that our algorithm for finding an $\alpha$-colored $s$-SPT nicely extends to finding a minimum weight solution. We essentially reduce this \textsc{Minimum Weight Color Constrained Shortest Path Tree} (\textsc{Min CC-SPT}) problem to finding a minimum cost maximum flow. Then by invoking an algorithm that solves \textsc{Minimum Cost Maximum Flow}, we obtain an  algorithm (Sec.~\ref{sec:min-cc-spt}) for \textsc{Min CC-SPT}. We also give a simpler algorithm to solve the problem for the special case when $q=2$ (Sec.~\ref{sec:minrb}). Our results are summarized in Table~\ref{table1}.

\begin{table}
\centering
\resizebox{\linewidth}{!}{
\begin{tabular}{|cc|cc|}
\hline
\multicolumn{2}{|c|}{\textbf{Graph Class}} & \multicolumn{2}{c|}{\textbf{Running Time}} \\ \hline
\multicolumn{1}{|c|}{\textbf{Edges}} & \textbf{Weights} & \multicolumn{1}{c|}{\textsc{CC-SPT}} & \textsc{Min CC-SPT} \\ \hline
\multicolumn{1}{|c|}{\multirow{5}{*}{Directed}} & Uniform & \multicolumn{1}{c|}{\multirow{3}{*}{\begin{tabular}[c]{@{}c@{}}$\mathcal{O}(T_{SSP}(n,m) + T_{MF}(\widehat{n},\widehat{m}))$\\ $\mathcal{O}(T_{SSP}(n,m))$ when $q=2$\end{tabular}}} & \multirow{3}{*}{\begin{tabular}[c]{@{}c@{}}$\mathcal{O}(T_{SSP}(n,m) + T_{MCMF}(\widehat{n},\widehat{m}))$\\ $\mathcal{O}(T_{SSP}(n,m)+T_{sort}(n))$ when $q=2$\end{tabular}} \\ \cline{2-2}
\multicolumn{1}{|c|}{} & \multirow{2}{*}{\begin{tabular}[c]{@{}c@{}}All cycles have\\ positive weight\end{tabular}} & \multicolumn{1}{c|}{} &  \\
\multicolumn{1}{|c|}{} &  & \multicolumn{1}{c|}{} &  \\ \cline{2-4} 
\multicolumn{1}{|c|}{} & \begin{tabular}[c]{@{}c@{}}No negative-\\ weight cycles\end{tabular} & \multicolumn{1}{c|}{NP-Complete} & NP-Complete \\ \cline{2-4} 
\multicolumn{1}{|c|}{} & Arbitrary & \multicolumn{1}{c|}{SPT may not exist} & SPT may not exist \\ \hline
\multicolumn{1}{|c|}{\multirow{3}{*}{Undirected}} & Uniform & \multicolumn{1}{c|}{\multirow{2}{*}{\begin{tabular}[c]{@{}c@{}}$\mathcal{O}(T_{SSP}(n,m) + T_{MF}(\widehat{n},\widehat{m}))$\\ $\mathcal{O}(T_{SSP}(n,m))$ when $q=2$\end{tabular}}} & \multirow{2}{*}{\begin{tabular}[c]{@{}c@{}}$\mathcal{O}(T_{SSP}(n,m) + T_{MCMF}(\widehat{n},\widehat{m}))$\\ $\mathcal{O}(T_{SSP}(n,m)+T_{sort}(n))$ when $q=2$\end{tabular}} \\ \cline{2-2}
\multicolumn{1}{|c|}{} & Positive & \multicolumn{1}{c|}{} &  \\ \cline{2-4} 
\multicolumn{1}{|c|}{} & Arbitrary & \multicolumn{1}{c|}{SPT may not exist} & SPT may not exist \\ \hline
\end{tabular}
}
\smallskip
\caption{Running time of our algorithms for \textsc{CC-SPT} and \textsc{Min CC-SPT} on $q$-colored $w$-weighted graphs with $n$ vertices and $m$ edges. Here, $T_{sort}(n)$ denotes the running time of an algorithm to sort $n$ numbers, 
$T_{SSP}(n,m)$ denotes the running time of a single source shortest path algorithm on a graph with $n$ vertices and $m$ edges and $T_{MF}(\widehat{n},\widehat{m})$ (resp. $T_{MCMF}(\widehat{n},\widehat{m})$) denotes the running time of an algorithm that computes an integral maximum flow (resp. integral minimum cost maximum flow) in a network with $\widehat{n}=\mathcal{O}(n+q)$ vertices and $\widehat{m}=\mathcal{O}(\min(m,nq))$ edges. When edge weights of $G$ are integral and polynomially-bounded, the current fastest single source shortest path and (minimum cost) maximum flow algorithms have near-linear (in the number of edges) running times.} 
\label{table1}
\end{table}

\subparagraph*{Arborescences in DAGs.}

We observe that out-forests of a DAG rooted at a given vertex form a matroid $M$. Therefore, instances of \textsc{CC-ARB} on DAGs where the required solution has only upper bound color constraints can be solved using \textsc{2-Matroid Intersection}. As a consequence, we obtain polynomial-time algorithms for instances of \textsc{CC-SPT} which have only upper bound color constraints. Moreover, a solution to a general instance of \textsc{CC-ARB} on a DAG may be seen as a largest common independent set of the matroid $M$ and a generalized partition matroid, leading to a polynomial-time algorithm for finding the same~\cite{BrezovecCG88}. However, our flow-based algorithms are faster.

\paragraph*{Undirected Graphs. }In all our algorithms, though we assume that the input is a digraph, these algorithms also work for undirected graphs by the standard transformation into equivalent digraphs if the edge weighting is positive: every undirected edge between $u$ and $v$ with weight $w_{uv}$ is replaced by two directed edges $(u,v)$ and $(v,u)$ each of weight $w_{uv}$. However, in the presence of negative-weight edges, this transformation will introduce negative-weight cycles in the resulting digraph even though the original undirected graph has none. Finding shortest paths is sophisticated even in uncolored undirected graphs with negative-weight edges (but no negative-weight cycles) and it involves a nonbipartite weighted matching computation \cite{AhujaMO93,Gabow83,Gabow85,Lawler76}. In undirected graphs with negative-weight cycles, subpaths of an undirected shortest path are not necessarily shortest paths and the set of shortest paths from a single source may not define a tree (as in the case of digraphs). Observe that in a zero-weighted undirected graph, every spanning tree is an $s$-SPT.

\paragraph*{Color-constrained shortest $s$-$t$-path. }Having studied the complexity of finding a color-constrained $s$-SPT, it is natural to explore the complexity of finding a color-constrained shortest $s$-$t$-path. We call this problem as \textsc{CC-SP}. The vertex-colored variant of \textsc{CC-SP} has been recently studied~\cite{BentertKN23}. We describe linear-time reductions between these problems and show hardness and tractability results for \textsc{CC-SP} in Section~\ref{sec:sp}.

\paragraph*{Preliminaries. } A matroid is a pair ${\cal M} = (E, {\cal F})$, where $E$ is the ground set and ${\cal F}$ is a family of subsets of $E$ (called independent sets) satisfying the following three properties (1) $\emptyset \in {\cal F}$, (2) if $I \in {\cal F}$ and $J \subseteq I$, then $J \in {\cal F}$, and (3) if $I, J \in {\cal F}$ with $|J| < |I|$, then there exists an element $e \in I \setminus J$, such that $J \cup \LR{e} \in {\cal F}$. 
We use the following two types of matroids in this paper. The set of all acyclic subgraphs (forests) of an undirected graph forms a matroid called
a \emph{graphic matroid} with the edge set of the graph as the ground set. 
Given a partition $E_1, \ldots, E_q$ of the ground set $E$ and ``capacity constraints'' $\alpha_1, \ldots, \alpha_q \in \mathbb{N}$, the collection of all subsets $S$ of $E$ such that for each $i \in [q]$, $|S \cap E_i| \leq \alpha_i$, forms a matroid called a \emph{partition matroid}. In addition, if one specifies ``lower bound constraints'' $\beta_1, \ldots, \beta_q \in \mathbb{N}$, the collection of all subsets $S$ of $E$ such that $\beta_i \leq |S \cap E_i| \leq \alpha_i$ for each $i \in [q]$, is called a \emph{generalized partition matroid}.

We refer to the book by Diestel~\cite{Diestel06} for standard graph-theoretic definitions and terminology not defined here. For a digraph $G$, let $V(G)$ and $E(G)$ denote the sets of its vertices and edges/arcs, respectively.  The {\em endpoints} of the edge $(u,v)$ are $u$ and $v$ with $u$ referred to as its \emph{tail} and $v$ referred to as its \emph{head}. We also say that $(u,v)$ is an {\em incoming edge} of $v$ or an edge {\em entering $v$} and an {\em outgoing edge} of $u$ or an edge {\em leaving $u$}.  If $(u,v)$ is an edge, then $u$ is said to be an {\em in-neighbor} of $v$ and $v$ is said to be an {\em out-neighbor} of $u$. For a vertex $v$, its {\em out-neighborhood}, denoted by $N^{+}(v)$, is the set of its out-neighbors and its {\em in-neighborhood}, denoted by $N^{-}(v)$, is the set of its in-neighbors. The {\em out-degree} $d^{+}(v)$ and {\em in-degree} $d^{-}(v)$ of a vertex $v$ are the sizes of its out-neighborhood and in-neighborhood, respectively. The notation for neighborhood and degree may be subscripted with the notation for the underlying graph to avoid any ambiguity. For a digraph $D$, its {\em underlying undirected graph} $U(D)$ is the undirected graph $G$ defined as $V(G)=V(D)$ and $E(G)=\{(u,v) : u,v \in V(D), (u,v) \in E(D) \text{ or } (v,u) \in E(D)\}$.

\paragraph*{Road Map. }
In Section~\ref{sec:spt-arb}, we prove important relationships between arborescences and shortest path trees. We also show how to solve \textsc{CC-SPT} using \textsc{CC-ARB}. In Section~\ref{sec:spt-arb-fast}, we give an algorithm to solve \textsc{CC-SPT} using maximum flow computation. In Section~\ref{sec:min-cc-spt}, we give an algorithm to solve \textsc{Min CC-SPT}. 
In Section~\ref{sec:sp}, we describe hardness and tractability results for \textsc{CC-SP}.  We state some interesting open problems in Section~\ref{sec:concl}. 

\section{Shortest Path Trees and Arborescences}
\label{sec:spt-arb}
In this section, we relate shortest path trees and arborescences. We begin by mentioning some standard definitions and observations. A digraph $T$ is called an {\em out-tree} if $U(T)$ is a tree and every vertex has in-degree at most one. Out-trees have exactly one vertex whose in-degree is $0$ and that vertex is called as the {\em root} of the out-tree. A digraph $T$ is called an {\em out-forest} if $U(T)$ is a forest and every vertex has in-degree at most one. 


\subsection{NP-Completeness of \textsc{CC-ARB} and \textsc{CC-SPT}}

Now we show that \textsc{CC-ARB} is \NP-complete even when restricted to instances where $\alpha_i=\alpha'_i$ for each $i \in [q]$. 

\begin{mytheorem}
   \label{thm:col-arb-hard}
   \textsc{CC-ARB} is \NP-complete.
\end{mytheorem}
\begin{proof}
    The membership of \textsc{CC-ARB} in \NP\ is easy to verify. To show \NP-hardness, we give a reduction from the \textsc{Hamiltonian Path} problem of determining the existence of a Hamiltonian path starting at a given vertex in a digraph. 
    Consider a digraph $D$ on $n$ vertices and let $s \in V(D)$.
    The construction of an instance $(D',\chi,s,\alpha)$ of \textsc{CC-ARB} from the instance $(D,s)$ of \textsc{Hamiltonian Path} is as follows. Let $D'$ denote the graph obtained from $D$ by adding a new vertex $t$ and adding an edge from every other vertex to $t$. Let $V(D')=\{s=v_1, \ldots, v_n, v_{n+1}=t\}$. 
    Let $\chi: E(D')\rightarrow [n]$ be the coloring function such that an edge $(v_i,v_j)$ is colored $i$. 
    Note that $E(D')\setminus E(D)$ has one edge of every color. 
    Let $\alpha$ be $([1,1],[1,1],\ldots, [1,1])$. 
    We claim that $D'$ has an $\alpha$-colored $s$-arborescence if and only if $D$ has a Hamiltonian path starting at $s$. Suppose $D$ has a Hamiltonian path $P$ starting at $s$ and ending at $x$. Then, by adding the edge $(x,t)$ to $P$, we obtain a Hamiltonian path $P'$ in $D'$. 
    Observe that $P'$ has exactly one edge of color $i$ for each $i \in[q]$. Since a Hamiltonian path is an arborescence, $P'$ is an $\alpha$-colored $s$-arborescence of $D'$. 

    Conversely, suppose $D'$ has an $\alpha$-colored $s$-arborescence $T$. Then, in $T$, the in-degree of $s$ is $0$ and the in-degree of every other vertex is $1$. Further, the out-degree of $t$ in $T$ is $0$ since $t$ has no out-neighbors in $D'$. Moreover, for each $i \in [n]$, the out-degree of $v_i$ in $T$ is $1$, since there is exactly one edge in $T$ that has color $i$. Therefore, $T$ is a Hamiltonian path of $D'$ starting from $s$ and ending at $t$. Then, the subpath of $T$ obtained by deleting $t$ is a Hamiltonian path of $D$ starting at $s$. 
\end{proof}

Since \textsc{CC-ARB} on an arbitrary digraph reduces to \textsc{CC-SPT} on the same digraph in which all edge weights are set to zero, we get the following result from Theorem \ref{thm:col-arb-hard}. 

\begin{mytheorem}
\label{thm:spt-q-col-NP-h}
\textsc{CC-SPT} is \NP-complete even on digraphs with no negative-weight cycles but having zero-weight cycles.
\end{mytheorem}  

We note that the reduction given in Theorem \ref{thm:col-arb-hard} does not give any hardness result if we restrict to DAGs (\textsc{Hamiltonian Path} on DAGs is polynomial-time solvable). 

\subsection{\textsc{CC-SPT} and \textsc{CC-ARB}}

Consider a $q$-colored $w$-weighted digraph $G$. 
Let $s \in V(G)$ be such that all vertices of $G$ are reachable from $s$. Let $\delta(u,v)$ denote the distance from $u$ to $v$ which is defined as the weight of a minimum weight (shortest) $u$-$v$-path in $G$. For a vertex $v$ unreachable from $u$, $\delta(u,v)=\infty$. The \textit{$s$-shortest paths graph} $G_s$ of $G$, which we define next, will play a crucial role in this article. 

\begin{definition}\label{defn:shortestpath_subgraph}
For a $w$-weighted digraph $G$ and a vertex $s \in V(G)$, the $s$-shortest paths graph ($s$-SPG) $G_s$ of $G$ is the subgraph of $G$ with $V(G_s)=V(G)$ and $E(G_s)= \{ (u,v) \in E(G) : \delta(s,v) = \delta(s,u) + w(u,v) \}$. If $G$ is edge-colored with a coloring function $\chi$, then $G_s$ is also edge-colored with the coloring function $\chi_{|_{E(G_s)}}$.
\end{definition}

Observe that $G_s$ may also have parallel edges. We now list some basic observations about $G_s$. 

\begin{myobservation}
\label{obs:SPS-DAG}
If $G$ is a $w$-weighted digraph with only positive cycles, then $G_s$ is acyclic.    
\end{myobservation}
\begin{proof}
Suppose $G_s$ contains a cycle $C=(v_1,v_2,\ldots,v_k,v_1)$. By the definition of $G_s$, we have  $\delta(s,v_1) - \delta(s,v_k)= w(v_k,v_1)$ and for $2 \le i \le k$, $\delta(s,v_i) - \delta(s,v_{i-1}) = w(v_{i-1}, v_i)$. Combining these equalities, we get $0 = w(v_k,v_1) + \sum_{i=2}^{k}{w(v_{i-1}, v_i)}=w(C)$, a contradiction to the assumption that $G$ has only positive-weight cycles.
\end{proof}

\begin{myobservation}
\label{obs:uv-SPS}
In a $w$-weighted digraph $G$ with only positive cycles, an edge $(u,v) \in E(G)$ is an edge in $G_s$ if and only if $(u,v)$ is the last edge of some shortest $s$-$v$ path in $G$. Moreover, every edge of $G$ that is in some shortest path from $s$ is in $G_s$.
\end{myobservation}
\begin{proof}
It is clear that if $(u,v) \in E(G)$ is the last edge of some shortest $s$-$v$ path in $G$, then $\delta(s,v) = \delta(s,u) + w(u,v)$, and hence $(u,v) \in E(G_s)$. Conversely, let $(u,v) \in E(G_s)$. Then, by definition of $G_s$, we have $\delta(s,v) = \delta(s,u) + w(u,v)$. Let $P$ be a shortest $s$-$u$ path in $G$ and let $P'$ be the walk obtained by extending $P$ with $(u,v)$. Then, $w(P')=w(P)+w(u,v) = \delta(s,u)+w(u,v)=\delta(s,v)$. If $P'$ is not a path, then it contains a subpath $P''$ from $s$ to $v$. Since every cycle in $G$ has a positive weight, we get $w(P'') < w(P') = \delta(s,v)$, a contradiction. Hence $(u,v)$ is the last edge of a shortest $s$-$v$ path, namely $P'$. 

Moreover, since a subpath of a shortest path is a shortest path (optimal substructure property) in a graph with no negative-weight cycles, every edge of $G$ that is in some shortest path from $s$ is in $G_s$.
\end{proof}

\begin{myobservation}
\label{obs:spt-out-tree}
In a $w$-weighted digraph $G$ with only positive cycles, $T$ is an $s$-SPT of $G$ if and only if $T$ is an $s$-arborescence of $G_s$. 
\end{myobservation}
\begin{proof}
Consider an $s$-arborescence $T$ of $G_s$. It is enough to show that for each vertex $v\in V(G_s)$, the $s$-$v$ path $P_v$ in $T$ is a shortest $s$-$v$ path in $G$. Assume on the contrary that there is some vertex $x$ such that $P_x$ is not a shortest $s$-$x$ path in $G$. Among all such vertices, choose a vertex $x$ with minimum number of edges in the path $P_x$. Clearly $x \neq s$. Let $z$ be the predecessor of $x$ in $P_x$. Then $w(P_x)=w(P_z)+w(z,x)=\delta(s,z)+w(z,x)$. As $(z,x) \in E(G_s)$, we have $\delta(s,x)=\delta(s,z)+w(z,x)$ implying that $P_x$ is a shortest $s$-$x$ path in $G$, a contradiction.   

Conversely suppose $T$ is an $s$-SPT of $G$. Clearly $T$ is an $s$-arborescence of $G$. By optimal substructure property of shortest paths, every edge in $T$ is in some shortest $s$-$v$ path in $G$. Hence, by Observation \ref{obs:uv-SPS}, $E(T) \subseteq E(G_s)$. Therefore, $T$ is also an $s$-arborescence of $G_s$.
\end{proof}

From Observation~\ref{obs:spt-out-tree} and Observation~\ref{obs:SPS-DAG}, the following connections between \textsc{CC-SPT} and \textsc{CC-ARB} are immediate corollaries. 
\begin{mycorollary}
\label{cor:spt-arb-connection}
In a $w$-weighted $q$-colored digraph $G$ with only positive cycles, $T$ is a (minimum weight) $\alpha$-colored $s$-SPT of $G$ if and only if $T$ is a (minimum weight) $\alpha$-colored $s$-arborescence of $G_s$. 
\end{mycorollary}


\subsection{\textsc{CC-ARB} in DAGs and \textsc{2-Matroid Intersection}} 

We show that \textsc{CC-ARB} on a DAG $G$ and a color-constraint vector $\alpha=([0,\alpha'_1],[0,\alpha'_2],\ldots,[0,\alpha'_q])$ is polynomial-time solvable by a reduction to \textsc{2-Matroid Intersection}. Let $s \in V(G)$. If $s$ is not a source in $G$, then there is no $s$-arborescence in $G$. Hence we assume that $s$ is a source vertex in $G$. We show that the collection of all out-forests of $G$ forms a matroid with $E(G)$ as the ground set. Let $\mathcal{I}_1$ denote the set of all subsets $S \subseteq E(G)$ such that the subgraph $H=(V(G),S)$ has maximum in-degree one. Then $M_1 = (E(G), \mathcal{I}_1)$ forms a partition matroid where the edges are partitioned based on their head. As any subgraph of a DAG in which the in-degree of every vertex is at most one is an out-forest, it follows that $\mathcal{I}_1$ is the collection of edge-sets of all out-forests of $G$.

Let $E_i$ denote the edges of $G$ that are colored using $i$. Now, define a second partition matroid $M_2=(E(G),\mathcal{I}_2)$ where $\mathcal{I}_2=\{S\subseteq E(G) : |S \cap E_i| \leq \alpha'_i \text{ for every } i\in [q]\}$. Then, an $\alpha$-colored $s$-arborescence of $G$ is a common independent set of $\mathcal{I}_1$ and $\mathcal{I}_2$ of size $|V(G)|-1$ and vice-versa. Hence, it follows that such an $s$-arborescence can be obtained by solving the corresponding instance of \textsc{2-Matroid Intersection}. 

Next, consider \textsc{CC-ARB} on a DAG $G$ and color-constraint vector $\alpha$ of the form $\alpha=([\alpha_1,\infty],[\alpha_2,\infty],\ldots,[\alpha_q,\infty])$. Let $G'$ denote the graph obtained from $G$ by duplicating every edge and coloring the new copies with a new color $q+1$. Then $G$ has an $s$-arborescence with at least $\alpha_i$ edges of color $i$ for each $i \in [q]$ if and only if $G'$ has an $s$-arborescence with at most $\alpha_i$ edges of color $i$ for each $i \in [q+1]$ where $\alpha_{q+1}=|V(G)|-1- \sum_{i=1}^{q} \alpha_i$. Therefore, a desired solution can be obtained by solving the corresponding instance of \textsc{2-Matroid Intersection}. 


Finally, consider a general instance of \textsc{CC-ARB} on a DAG $G$ and color-constraint vector $\alpha=([\alpha_1,\alpha'_1],\ldots,[\alpha_q,\alpha'_q])$. An $\alpha$-colored $s$-arborescence of $G$ may be seen as a largest common independent set of the matroid $M_1$ and the generalized partition matroid $M_3=(E(G),\mathcal{I}_3)$ where $\mathcal{I}_3=\{S \subseteq E(G) : \alpha_i \leq |S| \leq \alpha'_i, i \in [q]\}$, leading to a polynomial-time algorithm for finding the same~\cite{BrezovecCG88}. 

\section{Color-Constrained Shortest Path Trees} 
\label{sec:spt-arb-fast}
Consider a $q$-colored $w$-weighted digraph $G$ with only positive-weight cycles, a vertex $s \in V(G)$ and a color-constraint vector $\alpha=([\alpha_1,\alpha'_1],\ldots,[\alpha_q,\alpha'_q])\in (\mathbb{N} \times \mathbb{N})^q$. Let $n$ and $m$ denote the number of vertices and edges respectively in $G$. Observe that if $\sum_i \alpha_i >n-1$ or $\sum_i \alpha'_i <n-1$, then the instance has no solution and hence we may assume that $\sum_i \alpha_i \leq n-1$ and $\sum_i \alpha'_i \geq n-1$. Further, we may also assume that $\alpha'_i \leq n-1$ for each $i \in [q]$. 
Our algorithm for \textsc{CC-SPT} involves the following three steps.
\begin{itemize}
    \item Step~1: Use a single source shortest path algorithm on $(G,w,s)$ to compute $\delta(s,v)$ for each $v \in V(G)$.
    \item Step~2: Construct $G_s$ with $V(G_s)=V(G)$, $E(G_s)= \{ (u,v) \in E(G) : \delta(s,v) = \delta(s,u) + w(u,v) \}$. 
    \item Step~3: Use Algorithm~\textsf{CC-ARB-Flow} on $(G_s, \alpha)$, to find an $\alpha$-colored $s$-arborescence $T$ in $G_s$, if one exists. If a solution $T$ is found, return $T$. Otherwise, report that no solution exists.
\end{itemize}
By Observation~\ref{obs:SPS-DAG}, the graph $G_s$, computed in Step~2 is a DAG.  Using $G_s$ as the input DAG, Algorithm~\textsf{CC-ARB-Flow} computes an $\alpha$-colored $s$-arborescence $T$ of $G_s$, if one exists. From Corollary~\ref{cor:spt-arb-connection}, $T$ is an $\alpha$-colored $s$-SPT in $G$. Crucial to the correctness in the computation of an $\alpha$-colored $s$-arborescence of $G_s$ in Step~3 is the following observation.

\begin{myobservation}
\label{obs:spt-G}
A spanning subgraph $H$ of $G_s$ is an $s$-arborescence if and only if exactly one edge of $H$ enters each vertex $v \in V(H) \setminus \{s\}$. 
\end{myobservation}
\begin{proof}
The forward direction of the claim follows by the definition of $s$-arborescence. Now, for the converse, consider a spanning subgraph $H$ of $G_s$ such that for each $v \in V(H) \setminus \{s\}$, $d^-_H(v)=1$. Since any DAG has a vertex with in-degree zero, $d^-_H(s)=0$. Since $G_s$ is acyclic (Observation \ref{obs:SPS-DAG}), so is $H$.  We claim that $U(H)$ is also acyclic. If $U(H)$ has a cycle $C$, then the corresponding directed subgraph of $H$ (which is acyclic) has at least one vertex with in-degree at least two, which is a contradiction. Since $U(H)$ is acyclic and has exactly $|V(H)|-1$ edges, it follows that $H$ is an $s$-arborescence of $G_s$.
\end{proof}

Now, we describe Algorithm~\textsf{CC-ARB-Flow} that computes an $\alpha$-colored $s$-arborescence (if one exists) in an arbitrary $q$-colored DAG $G_s$.  Algorithm \textsf{CC-ARB-Flow} first computes the number $\pi(v,i)$ of edges colored $i$ entering $v$ for each $v\in V(G_s)\setminus \{s\}$ and $i \in [q]$. Then, it constructs an auxiliary flow network $H$ using this information and computes an integral maximum feasible flow in $H$. The value of this flow determines the existence of an $\alpha$-colored $s$-arborescence of $G_s$ and when one exists it is easily constructed from this flow. Crucial to the correctness of the algorithm is Lemma~\ref{lem:spt-maxflow}.



\begin{algorithm}[ht]
	\caption{\textsf{CC-ARB-Flow}$(G_s,\alpha)$}
	\label{alg:flow-arb}
	\begin{algorithmic}[1]
		\Require A $q$-colored DAG $G_s$ where $s$ is designated as root and $\alpha \in (\mathbb{N} \times \mathbb{N})^q$
		\Ensure An $\alpha$-colored $s$-arborescence of $G_s$ (if one exists)
        \State Let $E_i$ denote the edges of $G_s$ that are colored using $i$. For each $v\in V(G_s)\setminus \{s\}$ and $i \in [q]$, compute the number $\pi(v,i)$ of edges in $E_i$ entering $v$ using a single linear scan through $E(G_s)$  
        \blueLcomment{Construction of auxiliary flow network $H$}
        \State Initialize a flow network $H$ with $V(H)=\{s^*,t^*\} \cup C \cup U$ where $C=[q]$, $U=V(G_s) \setminus \{s\}$, source $s^*$, sink $t^*$ and edges $E(H)=\emptyset$ along with capacity function $c:E(H) \rightarrow \mathbb{Z}^{\geq 0}$ and demand function $d:E(H) \rightarrow \mathbb{Z}^{\geq 0}$
        \For {each $i \in C$}
            \State add $(s^*,i)$ to $E(H)$ and set $c(s^*,i)=\alpha'_i$, $d(s^*,i)=\alpha_i$ 
        \EndFor
        \For {each $i \in C$ and $v\in U$ with $\pi(v,i)>0$}
            \State add $(i,v)$ to $E(H)$ and set $c(i,v)=1$, $d(i,v)=0$
        \EndFor
        \For {each $v\in U$} 
            \State add $(v,t^*)$ to $E(H)$ and set $c(v,t^*)=1$, $d(v,t^*)=0$ 
        \EndFor
        \blueLcomment{Construction of solution $T$}
        \State Compute an integral feasible flow $f$ of value $n-1$ (if one exists) in $H$. \label{line:flow}
        \If{no such flow $f$ exists}
            \State \Return $\emptyset$ to convey that $G_s$ has no $\alpha$-colored $s$-arborescence
        \Else
            \State Initialize graph $T$ with $V(T)=V(G_s)$ and $E(T)=\emptyset$
            \For{each $v \in U$ and $i\in C$ such that $f(i,v)>0$}
                \State add an arbitrary edge $(x,v) \in E_i$ to $T$
            \EndFor
            \State \Return $T$
        \EndIf
	\end{algorithmic}
\end{algorithm}

\begin{mylemma}
\label{lem:spt-maxflow}
$H$ has a feasible flow of value $n-1$ if and only if $G_s$ has an $\alpha$-colored $s$-arborescence.
\end{mylemma}
\begin{proof}
Suppose $G_s$ has an $\alpha$-colored $s$-arborescence $T$. For each $i \in [q]$, let $\beta_i$ denote the number of edges of color $i$ in $T$. We define a function $f:V(H) \times V(H) \rightarrow \mathbb{Z}^{\geq 0}$ as follows. 
\begin{itemize}
    \item For each $i \in C$, $f(s^*,i)=\beta_i$. 
    \item For each $i \in C$ and $v \in U$, $f(i,v)=1$ if the edge entering $v$ in $T$ has color $i$. 
    \item For each $v\in U$, $f(v,t^*)=1$.
    \item $f$ is $0$ everywhere else.
\end{itemize}

It is easy to verify that $f$ respects non-negativity, capacity and demand constraints of $H$. Consider a vertex $v \in U$. As $v$ has exactly one edge of $T$ entering it, we have $\sum_x f(x,v)=1$. In $H$, $v$ has exactly one edge $(v, t^*)$ leaving it and hence $\sum_x f(v,x)=f(v,t^*)=1$. Hence, $f$ satisfies flow conservation constraints at every vertex $v \in U$. Next, consider a vertex $i \in C$. Since $T$ has $\beta_i$ edges of color $i$, we have $\sum_v f(i,v) = \beta_i$. As $f(s^*,i) = \beta_i$, $f$ satisfies flow conservation constraints at every vertex $i \in C$. Finally, as $\sum_i f(s^*,i)= \sum_i \beta_i =n-1$, $f$ has value $n-1$. 

Conversely, suppose $H$ has a feasible flow $f$ of value $n-1$. Then, $f$ is a maximum flow of $H$, as the sum of capacities of edges entering $t^*$ is $n-1$. Since all the capacities are integral, we can assume $f$ to be integral. Let $v$ be any vertex in $U$.  Since the total flow entering $t^*$ is $n-1$, we have $f(v, t^*) = c(v, t^*) = 1$. To conserve the flow at $v$, we must have $f(i,v)=1$ for exactly one $i \in C$. Let $i$ be any vertex in $C$. To conserve the flow at $i$, we have $\sum_v f(i,v) = f(s^*,i) \leq c(s^*, i) = \alpha'_i$. Similarly, $\sum_v f(i,v) = f(s^*,i) \geq d(s^*, i) = \alpha_i$. Now, construct a graph $T$ with $V(T)=V(G_s)$ as follows: for each $v \in U$ and $i\in C$ such that $f(i,v) > 0$, add an arbitrary $i$-colored edge in $E(G_s)$ entering $v$ to $T$. Then, by properties of $f$ established above, for each $v \in U$, there is exactly one edge in $T$ entering $v$ and for each color $i$, there are at most $\alpha'_i$ and at least $\alpha_i$ edges of color $i$ in $T$. 
By Observation \ref{obs:spt-G}, $T$ is an $\alpha$-colored $s$-arborescence.
\end{proof}

Algorithm \textsf{CC-ARB-Flow} (in Line~\ref{line:flow}) computes an integral feasible flow of value $n-1$ (if one exists) in $H$. This can be done by a linear-time reduction to the problem of finding an integral maximum flow in another network as given by the following lemma. Note that reductions with similar constructions are known in the literature (see for example~\cite{Erickson15}).





\begin{restatable}[]{mylemma}{feasibleflow}
\label{lem:feasible-flow-max-flow}
    There is a linear-time algorithm that given a flow network $(H,s,t,c,d)$ with capacity function $c$, demand function $d$ and an integer $k$ produces another flow network $(H_k,s',t',c')$ with capacity function $c'$ (and no demand function) such that $H$ has a feasible flow of value $k$ if and only if $H_k$ has a flow of value $k+\sum_{(u,v) \in E(H)} d(u,v)$. Further, one can transform a flow $f'$ in $H_k$ of value $k+\sum_{(u,v) \in E(H)} d(u,v)$ into a feasible flow $f$ in $H$ of value $k$ in linear time. 
\end{restatable}
\begin{proof}
     \textbf{(sketch)} $H_k$ is constructed from $H$ by adding a new source vertex $s'$ and a new sink vertex $t'$ such that there is an edge from $s'$ to every vertex in $V(H)$ and an edge from every vertex in $V(H)$ to $t'$. 
     The capacity function $c' : V(H_k) \times V(H_k) \rightarrow \mathbb{R}^{\geq 0}$ is defined as follows.
$$
c'(u,v)= 
\begin{cases}
\underset{x\in N^{-}(v)}\sum d(x,v) & \text{if } u=s' \text{ and } v\in V(H)\setminus \{s\} \\
\underset{y\in N^{+}(v)}\sum d(u,y) & \text{if } v=t' \text{ and } u\in V(H)\setminus \{t\} \\
c(u,v) - d(u,v) & \text{if }(u, v) \in E(H) \\
k & \text{if }(u,v)= (s', s) \text{ or } (u, v) = (t, t')
\end{cases}
$$
We claim that $H$ has a feasible flow of value $k$ if and only if $H_k$ has a flow of value $k+\sum_{e \in E(H)} d(e)$. Moreover, if $f'$ is a flow in $H_k$ of value $k+\sum_{e \in E(H)} d(e)$, then the function $f: E(H)\rightarrow \mathbb{R}^{\geq 0}$ defined as $f(u,v)= f'(u,v)+ d(u,v) \text{ for all }(u,v)\in E(H)$ is a flow of value $k$ in $H$. The complete proof is deferred to the appendix.  
\end{proof}

 Now, we have the following result. 

\begin{mytheorem}
\label{thm:arb-q-col-flow}
Algorithm~\textsf{CC-ARB-Flow} solves \textsc{CC-ARB} on $q$-colored DAGs with $n$ vertices and $m$ edges in $\mathcal{O}(m + T_{MF}(\widehat{n},\widehat{m}))$.
\end{mytheorem}  
\begin{proof}
Consider an instance $(G_s,\alpha)$ of \textsc{CC-ARB}. The algorithm computes $\pi(v,i)$ for each $v \in V(G_s) \setminus \{s\}$ and $i \in [q]$ in $\mathcal{O}(n+m)$ time which in turn is used in constructing the network $H$ that has $\widehat{n}=n+q+1$ vertices and $\widehat{m}=\mathcal{O}(\min(m,nq)+q+n)$ edges in $\mathcal{O}(\widehat{n}+\widehat{m})$ time. Then, we construct the network $H_{n-1}$ using Lemma~\ref{lem:feasible-flow-max-flow} in $\mathcal{O}(\widehat{n}+\widehat{m})$ time. We compute a maximum flow $f'$ in $H_{n-1}$ in $T_{MF}(\widehat{n},\widehat{m})$ time. If the value of $f'$ is less than $n-1+\sum_{i=1}^n \alpha_i$, then $G_s$ has no $\alpha$-colored $s$-arborescence (by Lemma \ref{lem:spt-maxflow} and Lemma~\ref{lem:feasible-flow-max-flow}) and we declare the same. Subsequently, we assume that the value of $f'$ is $n-1+\sum_{i=1}^n \alpha_i$. If $f'$ is not integral, since the capacities and demands in $H$ are integral, we obtain an integral flow $\widetilde{f}$ of the same value (and same cost) using the rounding algorithm by Kang and Payor~\cite{KangP15} in $\mathcal{O}(\widehat{m} \log \widehat{m})$ time. Then, using Lemma~\ref{lem:feasible-flow-max-flow}, we obtain an integral feasible flow $f$ of value $n-1$ in $H$ from $\widetilde{f}$ in $\mathcal{O}(\widehat{n}+\widehat{m})$ time. Finally, by Lemma \ref{lem:spt-maxflow}, we construct an $\alpha$-colored $s$-arborescence of $G_s$ in $\mathcal{O}(\widehat{m})$ time. Therefore, the overall running time of the algorithm is $\mathcal{O}(m + T_{MF}(\widehat{n},\widehat{m}))$. 
\end{proof}

The correctness of the algorithm for \textsc{CC-SPT}, described in the beginning of this section, follows by Observation~\ref{obs:SPS-DAG} and Corollary~\ref{cor:spt-arb-connection}. The running time of Step~1 of the algorithm is $\mathcal{O}(T_{SSP}(n,m))$. Step~2 of the algorithm takes $\mathcal{O}(n+m)$ time. By Theorem~\ref{thm:arb-q-col-flow}, Step~3 takes $\mathcal{O}(m+T_{MF}(\widehat{n},\widehat{m}))$ time. Thus, we have the following result.

\begin{mytheorem}
\label{thm:spt-q-col-flow}
\textsc{CC-SPT} on $q$-colored weighted digraphs $G$ with $n$ vertices, $m$ edges and only positive cycles can be solved in $\mathcal{O}(T_{SSP}(n,m) + T_{MF}(\widehat{n},\widehat{m}))$ time.
\end{mytheorem}
\begin{remark}\label{rmk:flow-sssp}
By using the algorithm by van den Brand et al.~\cite{Brand0PKLGSS23} for \textsc{Minimum Cost Maximum Flow} for computing single source shortest path distances in $G$ and for solving \textsc{CC-ARB} in $G_s$, it follows that \textsc{CC-SPT} can be solved in $\mathcal{O}(m+\widehat{m}^{1+o(1)})$ time when the edge weights are integral and polynomially bounded in the size of the graph provided $G$ has only positive cycles. 
\end{remark}

\subsection{CC-ARB in Red-Blue Graphs}
\label{sec:rb}
In this section, we give an algorithm for \textsc{CC-ARB} on DAGs in the special case when $q=2$. We refer to a $2$-colored graph as a red-blue graph with the natural interpretation that red and blue are the colors. Since a color constraint $([\alpha_1,\alpha'_1],[\alpha_2,\alpha'_2])$ is equivalent to $([0,\min(\alpha_1,n-1-\alpha'_2)],[0,\min(\alpha_2,n-1-\alpha'_1)])$, we assume that there are constraints only on the maximum number of edges of each color. 

Consider a red-blue DAG $G_s$ and a color-constraint vector $\alpha=(\alpha_1,\alpha_2) \in \mathbb{N}^2$ where we desire $([0,\alpha_1],[0,\alpha_2])$-colored $s$-arborescence. We describe Algorithm \textsf{CC-RB-ARB} that outputs an $\alpha$-colored $s$-arborescence (if one exists) of $G_s$ in $\mathcal{O}(n+m)$ time. The algorithm is based on the observation that vertices in $V(G_s) \setminus \{s\}$ without red (resp. blue) edges entering them in $G_s$ must have a blue (resp. red) edge entering them in every solution and vertices with edges of both colors entering them in $G_s$ can have an edge of any of these colors entering them in a solution. Based on this observation, we partition the vertices into three classes and we show that the existence of an $\alpha$-colored $s$-arborescence can be determined by examining this partition. The remaining part of the algorithm computes such a solution (if one exists).

\begin{algorithm}[ht]
	\caption{\textsf{CC-RB-ARB}$(G_s,\alpha)$}
	\label{alg:spt-2-cols}
	\begin{algorithmic}[1]
	    \Require A red-blue DAG $G_s$ with $s$ designated as the root and $\alpha \in \mathbb{N}^2$
		\Ensure An $\alpha$-colored $s$-arborescence of $G_s$ (if one exists)
        \State Let $E_i$ denote the edges of $G_s$ that are colored using $i$. For each $v\in V(G_s)\setminus \{s\}$ and $i \in [2]$, compute the number $\pi(v,i)$ of edges in $E_i$ entering $v$ using a single linear scan through $E(G_s)$  	
        \blueLcomment{Construction of a partition of $V(G_s)\setminus \{s\}$}
        \State Partition the vertices of $V(G_s)\setminus \{s\}$ into $V_R$, $V_B$ and $V_{RB}$ as
        \State \hspace{.5cm} $V_R=\{v\in V(G_s)\setminus \{s\} : \pi(v,1)>0, \pi(v,2)=0 \}$
        \State \hspace{.5cm} $V_B=\{v\in V(G_s)\setminus \{s\} : \pi(v,1)=0, \pi(v,2)>0 \}$
        \State \hspace{.5cm} $V_{RB}=\{v\in V(G_s)\setminus \{s\} : \pi(v,1),\pi(v,2)>0 \}$
        \blueLcomment{Phase 3: construction of solution $T$}
        \If{$|V_R|>\alpha_1$ or $|V_B|>\alpha_2$}
            \State \Return $G$ has no $\alpha$-colored $s$-arborescence
        \Else
         \If {$|V_{RB}| > \alpha_1-|V_R|+\alpha_2-|V_B|$}
            \State \Return $G_s$ has no $\alpha$-colored $s$-arborescence
          \Else  
              \State Partition $V_{RB}$ into parts $X$ and $Y$ such that $|X|\leq \alpha_1-|V_R|$ and $|Y|\leq \alpha_2-|V_B|$ 
        \State Initialize graph $T$ with $V(T)=V(G_s)$ and $E(T)=\emptyset$
        \For {each $v \in V_R \cup X$}
            \State add an arbitrary red edge of $G_s$ entering $v$ to $T$
        \EndFor
        \For {each $v \in V_B \cup Y$} 
            \State add an arbitrary blue edge of $G_s$ entering $v$ to $T$
        \EndFor
        \State \Return $T$
        \EndIf
        \EndIf
	\end{algorithmic}
\end{algorithm}

\begin{mytheorem}
 \label{cc-spt-2cols}   
Algorithm \textsf{CC-RB-ARB} solves \textsc{CC-ARB} on red-blue DAGs with $n$ vertices and $m$ edges in $\mathcal{O}(n+m)$ time. 
\end{mytheorem}

\begin{proof}
Consider a red-blue DAG $G_s$ and a color-constraint vector $\alpha \in \mathbb{N}^2$. The algorithm first partitions (in linear time) vertices of $V(G_s)\setminus\{s\}$ into three sets $V_R$, $V_B$ and $V_{RB}$. Then, it partitions (in linear time) $V_{RB}$ into $X$ and $Y$ such that $|X| \leq \alpha_1-|V_R|$ and $|Y| \leq \alpha_2-|V_B|$. Observe that if $G_s$ has an $\alpha$-colored $s$-arborescence $T^*$, then $|V_R|\leq \alpha_1$ and $|V_B|\leq \alpha_2$. Further, at most $\alpha_1-|V_R|$ vertices $X^*$ of $V_{RB}$ have a red edge entering them and at most $\alpha_2-|V_B|$ vertices $Y^*$ of $V_{RB}$ have a blue edge entering them in $T^*$. As $T^*$ is a spanning subgraph of $G$, it follows that $X^*$ and $Y^*$ partition $V_{RB}$. 
Therefore, if $|V_{RB}| > \alpha_1-|V_R| + \alpha_2-|V_B|$, then the input instance is a \no-instance. Otherwise, it is possible to take at most $\alpha_1-|V_R|$ arbitrary vertices of $V_{RB}$ into $X$ and the rest into $Y$. 

Finally, a subgraph $T$ of $G_s$ is constructed (in linear time) in such a way that for each vertex $v\in V(G_s)\setminus\{s\}$, there is exactly one edge entering $v$. Clearly, $T$ has at most $\alpha_1$ red edges and at most $\alpha_2$ blue edges. From Observation \ref{obs:spt-G}, $T$ is $s$-arborescence. 
\end{proof}

Observe that Theorem~\ref{cc-spt-2cols} implies that \textsc{CC-SPT} on $2$-colored weighted digraphs $G$ with $n$ vertices, $m$ edges and only positive cycles can be solved in $\mathcal{O}(T_{SSP}(n,m))$ time. 

\section{Minimum Weight Color-Constrained Shortest Path Trees}
\label{sec:min-cc-spt}
In this section, we describe an algorithm for the \textsc{Minimum Weight Color Constrained Shortest Path Tree} (\textsc{Min CC-SPT}) problem of finding an $\alpha$-colored $s$-SPT of minimum total weight. The key subroutine is an algorithm for \textsc{Minimum Weight Color-Constrained Arborescence} (\textsc{Min CC-ARB}) on DAGs where given a $q$-colored DAG with a weight function on its edges, a source vertex $s$ and a color-constraint vector $\alpha \in (\mathbb{N} \times \mathbb{N})^q$, the objective is to find an $\alpha$-colored $s$-arborescence of minimum total weight.  

Consider a $q$-colored digraph $G$ with an associated weighting $w$ on its edges such that $G$ has only positive-weight cycles, a vertex $s \in V(G)$ and a color-constraint vector $\alpha=([\alpha_1,\alpha'_1],\ldots,[\alpha_q,\alpha'_q])\in (\mathbb{N} \times \mathbb{N})^q$. Like our algorithm for \textsc{CC-SPT}, the algorithm for \textsc{Min CC-SPT} involves the following three steps.
\begin{itemize}
    \item Step~1: Use a single source shortest path algorithm on $(G,w,s)$ to compute $\delta(s,v)$ for each $v \in V(G)$.
    \item Step~2: Construct $G_s$ with $V(G_s)=V(G)$, $E(G_s)= \{ (u,v) \in E(G) : \delta(s,v) = \delta(s,u) + w(u,v) \}$. If $G$ is $w$-weighted, then $G_s$ is $w'$-weighted where $w'=w_{|_{E(G_s)}}$. 
    \item Step~3: Use Algorithm~\textsf{Min-CC-ARB-Flow} on $(G_s, w', \alpha)$, to find a minimum cost $\alpha$-colored $s$-arborescence $T$ in $G_s$, if one exists. If a solution $T$ is found, return $T$. Otherwise, report that no solution exists.
\end{itemize}

By Observation~\ref{obs:SPS-DAG}, the graph $G_s$, computed in Step~2 is a DAG. Using $G_s$ as the input DAG, Algorithm~\textsf{Min-CC-ARB-Flow} computes a minimum weight $\alpha$-constrained $s$-arborescence $T$ of $G_s$, if one exists. It does so by casting the problem as a minimum cost maximum flow problem. Then, from Observation~\ref{obs:spt-out-tree}, $T$ is a minimum weight $\alpha$-colored $s$-SPT in $G$. Crucial to the correctness of Algorithm \textsf{Min-CC-ARB-Flow} is the following result on the structure of a solution. 

\begin{mylemma}
\label{lemma:minWt_predecessor}
Let $G_s$ be a $q$-colored DAG with a weight function $w: E(G) \rightarrow \mathbb{R}$ on its edges and $T$ be an $\alpha$-colored $s$-arborescence of minimum total weight in $G_s$. For each $i \in [q]$ and $v \in V(T) \setminus \{s\}$, if $(u,v) \in E(T)$ is an $i$-colored edge entering $v$, then $(u,v)$ is an edge that minimizes $w(u,v)$ over all $i$-colored edges entering $v$.
\end{mylemma}
\begin{proof}
    Consider $v \in V(G_s) \setminus \{s\}$ and let $(u,v)\in E(T)$ be an $i$-colored edge. 
    Assume on the contrary that there is an $i$-colored edge $(x,v)$ entering $v$ such that $w(x,v)<w(u,v)$. Since $T$ is an arborescence, $(x,v) \notin E(T)$. Let $T_u$ and $T_v$ denote the subtrees of $T-(u,v)$ containing $u$ and $v$, respectively.
    Note that $v$ is the root in $T_v$.
    If $x \in V(T_v)$, then it follows that $G_s$ has a cycle contradicting the fact that $G_s$ is a DAG. Therefore, $x \in T_u$. Then, the arborescence $T-(u,v)+(x,v)$ is an $\alpha$-colored $s$-arborescence with a weight strictly smaller than that of $T$ leading to a contradiction.  
\end{proof}

Now, we describe Algorithm~\textsf{Min-CC-ARB-Flow} which is similar to Algorithm~\textsf{CC-ARB-Flow}. 
The key difference from Algorithm~\textsf{CC-ARB-Flow} is that we assign a cost (in addition to the capacity and demand) to every edge of the auxiliary flow network $H$ constructed and seek a minimum cost maximum flow instead of an arbitrary maximum flow. 

\begin{algorithm}[ht]
	\caption{\textsf{Min-CC-ARB-Flow}$(G_s, w, \alpha)$}
	\label{alg:flow-arb-min}
	\begin{algorithmic}[1]
		\Require A $q$-colored DAG $G_s$ where $s$ is designated as root, a weight function $w:E(G_s) \rightarrow \mathbb{R}$ and $\alpha \in (\mathbb{N} \times \mathbb{N})^q$
		\Ensure A minimum weight $\alpha$-colored $s$-arborescence of $G_s$ (if one exists)
        \State Let $E_i$ denote the edges of $G_s$ that are colored using $i$. For each $v\in V(G_s)\setminus \{s\}$ and $i \in [q]$, compute the number $\pi(v,i)$ of edges in $E_i$ entering $v$ using a single linear scan through $E(G_s)$   \label{line:pi}
        \blueLcomment{Construction of auxiliary flow network $H$}
        \State Initialize a flow network $H$ with $V(H)=\{s^*,t^*\} \cup C \cup U$ where $C=[q]$ and $U=V(G_s) \setminus \{s\}$, source $s^*$, sink $t^*$ and edges $E(H)=\emptyset$ along with capacity function $c:E(H) \rightarrow \mathbb{Z}^{\geq 0}$, demand function $d:E(H) \rightarrow \mathbb{Z}^{\geq 0}$ and cost function $\widehat{c}:E(H) \rightarrow \mathbb{R}^{\geq 0}$
        \For {each $i \in C$}
            \State add $(s^*,i)$ to $E(H)$ and set $c(s^*,i)=\alpha'_i$, $d(s^*,i)=\alpha_i$, $\widehat{c}(s^*,i)=0$
        \EndFor
        \For {each $i \in C$ and $v\in U$ with $\pi(v,i)>0$}
            \State add $(i,v)$ to $E(H)$ and set $c(i,v)=1$, $d(i,v)=0$, $\widehat{c}(i,v)= \min \{ w(u,v) : (u,v) \in E_i \cap E(G_s)\}$ \label{line:cost-min-wt}
        \EndFor
        \For {each $v\in U$} 
            \State add $(v,t^*)$ to $E(H)$ and set $c(v,t^*)=1$, $d(v,t^*)=0$, $\widehat{c}(v,t^*)=0$
        \EndFor
        \blueLcomment{Construction of a solution $T$}
        \State Compute an integral minimum cost feasible flow $f$ of value $n-1$ (if one exists) in $H$. 
        \If{no such flow $f$ exists}
            \State \Return $\emptyset$ to convey that $G_s$ has no $\alpha$-colored $s$-arborescence
        \Else
            \State Initialize graph $T$ with $V(T)=V(G_s)$ and $E(T)=\emptyset$
            \For{each $v \in U$ such that $f(i,v)>0$ for some $i\in C$}
                \State add an edge $(x,v)$ with $w(x,v)=\min \{ w(u,v) : (u,v) \in E_i \cap E(G_s)\}$ to $T$
            \EndFor
            \State \Return $T$
        \EndIf
	\end{algorithmic}
\end{algorithm}

\begin{mylemma}
\label{lem:min-arb-q-col-flow}
Algorithm~\textsf{Min-CC-ARB-Flow} solves \textsc{Min CC-ARB} on $q$-colored edge-weighted DAGs with $n$ vertices and $m$ edges in $\mathcal{O}(m+T_{MCMF}(\widehat{n},\widehat{m}))$ time.
\end{mylemma}  

\begin{proof}
First, we will discuss the correctness of the algorithm.  Consider an instance $(G_s, w, \alpha)$ of \textsc{Min CC-ARB}. As in \textsc{CC-ARB} algorithm, the flow graph $H$ cannot have a flow of value larger than $n-1$ and by Lemma~\ref{lem:spt-maxflow}, $H$ has a feasible flow $f$ of value $n-1$ if and only if $G_s$ has an $\alpha$-colored $s$-arborescence. Therefore, it remains to argue that if $|f|=n-1$, then $T$ is a minimum weight $\alpha$-colored $s$-arborescence of $G_s$. Observe that, by construction of $T$, the weight of $T$ is equal to the cost of the  flow $f$.  If $T$ is not a minimum weight $\alpha$-colored $s$-arborescence, then there is an $\alpha$-colored $s$-arborescence $T'$ with weight strictly smaller than $T$. By Lemma~\ref{lemma:minWt_predecessor}, for each $i \in [q]$ and $v \in V(T')$, if $(u,v) \in E(T')$ is an $i$-colored edge entering $v$, then $w(u,v)$ has the minimum weight over all $i$-colored edges entering $v$. Hence, in Line~\ref{line:cost-min-wt} of the algorithm, for each edge $(u, v)$ in $T'$, an edge of capacity $1$ and cost $w(u,v)$ from vertex $i$ to vertex $v$ gets added to $H$. These edges guarantee a flow $f'$ in $H$ with value $n-1$ and cost equal to the weight of $T'$. Since the weight of $T'$ is less than that of $T$ which in turn equal to the cost of $f$, this leads to a contradiction to the minimality of $f$. Therefore, $T$ is a minimum weight $\alpha$-colored $s$-arborescence of $G_s$. 

Now, we analyze the running time of the algorithm. Line~\ref{line:pi} computes $\pi(v,i)$ for each $v \in V(G_s) \setminus \{s\}$ and $i \in [q]$ in $\mathcal{O}(n+m)$ time which in turn is used in constructing the flow network $H$ that has $\widehat{n}=\mathcal{O}(n+q)$ vertices and $\widehat{m}=\mathcal{O}(\min(m,nq))$ edges in $\mathcal{O}(\widehat{n}+\widehat{m})$ time. Then, we construct the network $H_{n-1}$ using Lemma~\ref{lem:feasible-flow-max-flow} in $\mathcal{O}(\widehat{n}+\widehat{m})$ time. We compute an integral minimum cost maximum flow $f'$ in $H_{n-1}$  in $T_{MCMF}(\widehat{n},\widehat{m})$ time. If the value of $f'$ is less than $n-1+\sum_{i=1}^n \alpha_i$, then $G_s$ has no $\alpha$-colored $s$-arborescence (by Lemma \ref{lem:spt-maxflow} and Lemma~\ref{lem:feasible-flow-max-flow}) and we declare the same. Otherwise, using Lemma~\ref{lem:feasible-flow-max-flow} and Remark~\ref{rmk:generalization}, we obtain a minimum cost feasible flow $f$ of value $n-1$ in $H$ in $\mathcal{O}(\widehat{n}+\widehat{m})$ time. Finally, by Lemma~\ref{lem:spt-maxflow}, we construct an $\alpha$-colored $s$-arborescence of $G_s$ in $\mathcal{O}(\widehat{m})$ time. Therefore, the overall running time of the algorithm is $\mathcal{O}(m+T_{MCMF}(\widehat{n},\widehat{m}))$.
\end{proof}

The correctness of the algorithm for \textsc{Min CC-SPT} follows by Observation~\ref{obs:SPS-DAG} and Corollary~\ref{cor:spt-arb-connection}. The running time of Step~1 of the algorithm is $\mathcal{O}(T_{SSP}(n,m))$. Step~2 of the algorithm takes $\mathcal{O}(n+m)$ time. By Lemma~\ref{lem:min-arb-q-col-flow}, Step~3 takes $\mathcal{O}(m+T_{MCMF}(\widehat{n},\widehat{m}))$ time. Thus, we have the following result. 

\begin{mytheorem}
\label{thm:spt-q-col-minwt}
There is an algorithm that solves \textsc{Min CC-SPT} on $q$-colored edge-weighted digraphs $G$ with $n$ vertices, $m$ edges and only positive cycles in $\mathcal{O}(T_{SSP}(n,m)+T_{MCMF}(\widehat{n},\widehat{m})$ time where $T_{MCMF}(\widehat{n},\widehat{m})$ denotes the running time of an algorithm that computes an integral minimum cost maximum flow in a network with $\widehat{n}$ vertices and $\widehat{m}$ edges. 
\end{mytheorem}  



\subsection{Min CC-ARB in Red-Blue Graphs}
\label{sec:minrb}
In this section, we give an algorithm \textsf{Min-CC-RB-ARB} for the \textsc{Min CC-ARB} problem in the special case when $q=2$. This is a modification of Algorithm \textsf{CC-RB-ARB} given in Section \ref{sec:rb}. The key difference from \textsf{CC-RB-ARB} is a careful partitioning of $V_{RB}$ (the set of vertices which have incoming edges of both colors in $G_s$) based on the minimum weights of the incoming edges of each color.
   
\begin{algorithm}[hbt!]
	\caption{\textsf{Min-CC-RB-ARB}$(G_s,w,\alpha)$}
	\label{alg:min-arb-2-cols}
	\begin{algorithmic}[1]
 	      \Require A $w$-weighted red-blue DAG $G_s$ where $s$ is designated as the root and $\alpha \in \mathbb{N}^2$
		\Ensure A minimum weight $\alpha$-colored $s$-arborescence of $G_s$ (if one exists) 
         \State Let $E_i$ denote the edges of $G_s$ that are colored using $i$. For each $v\in V(G_s)\setminus \{s\}$ and $i \in [2]$, compute the number $\pi(v,i)$ of edges in $E_i$ entering $v$ using a single linear scan through $E(G_s)$   
        \blueLcomment{Partitioning $V(G_s) \setminus \{s\}$} 
        \State Partition $V(G_s)\setminus \{s\}$ into $V_R$, $V_B$ and $V_{RB}$ as
        \State \hspace{.5cm} $V_R=\{v\in V(G_s)\setminus \{s\} : \pi(v,1)>0, \pi(v,2)=0 \}$
        \State \hspace{.5cm} $V_B=\{v\in V(G_s)\setminus \{s\} : \pi(v,1)=0, \pi(v,2)>0 \}$
        \State \hspace{.5cm} $V_{RB}=\{v\in V(G_s)\setminus \{s\} : \pi(v,1)>0, \pi(v,2)>0 \}$
        \If{$\alpha_1 + \alpha_2 < n-1$ or $|V_R|>\alpha_1$ or $|V_B|>\alpha_2$}  \label{step:NoSol}
            \State \Return $G_s$ has no $\alpha$-colored $s$-arborescence
        \EndIf
        \For {each $v \in V_{RB}$}
            \State $r_v =$ the minimum weight of a red edge of $G_s$ entering $v$
            \State $b_v =$ the minimum weight of a blue edge of $G_s$ entering $v$
        \EndFor
        \State Partition $V_{RB}$ in to $V'_R$ and $V'_B$ as
            \State \hspace{.5cm} $V'_R = \{v \in V_{RB} : r_v \leq b_v\}$
            \State \hspace{.5cm} $V'_B = \{v \in V_{RB} : r_v > b_v\}$
        \If {$|V_R \cup V'_R| > \alpha_1$} \label{step:ReduceR}
            \State Sort the vertices in $V'_R$ in the non-decreasing order 
            of $(b_v - r_v)$
            \State Fix $S$ as the set of first $|V_R \cup V'_R| - \alpha_1$ vertices in the above order
        \ElsIf{$|V_B \cup V'_B| > \alpha_2$}  \label{step:ReduceB}
            \State Sort the vertices in $V'_B$ in the non-decreasing order 
            of $(r_v - b_v)$
            \State Fix $S$ as the set of first $|V_B \cup V'_B| - \alpha_2$ vertices in the above order
        \Else
            \State $S=\emptyset$ 
        \EndIf
        \blueLcomment{Construction of solution $T$}
        \State Initialize graph $T$ with $V(T)=V(G_s)$ and $E(T)=\emptyset$
        \For {each $v \in V_R \cup (V'_R \vartriangle S)$}
            \State add a minimum weight red edge $(x,v) \in E(G_s)$ to $T$
        \EndFor
        \For {each $v \in V_B \cup (V'_B \vartriangle S)$} 
            \State add a minimum weight blue edge $(x,v) \in E(G_s)$ to $T$
        \EndFor
        \State \Return $T$
	\end{algorithmic}
\end{algorithm}

Now, we proceed to the correctness and running time analysis of Algorithm \textsf{Min-CC-RB-ARB}. Let $T_{Sort}(n)$ denote the running time of the fastest sorting algorithm that sorts $n$ weights.  

\begin{mytheorem}\label{thm:minWt-red-blue}
Algorithm \textsf{Min-CC-RB-ARB} solves \textsc{Min CC-ARB} in $\mathcal{O}(m+T_{Sort}(n))$ time on red-blue edge-weighted DAGs with $n$ vertices and $m$ edges.
\end{mytheorem}

\begin{proof}
Consider a red-blue $w$-weighted DAG $G_s$ and a color-constraint vector $\alpha \in \mathbb{N}^2$. The running time is dominated by the sorting subroutines. Hence the overall running time is as claimed.

It is clear that there is no $\alpha$-colored $s$-arborescence in $G_s$ when any of the conditions in Line~\ref{step:NoSol} is true. Let $T$ be a solution returned by the algorithm.  It is easy to check that $T$ is an $\alpha$-colored $s$-arborescence. We show next that $T$ has the minimum weight among all $\alpha$-colored $s$-arborescence in $G_s$. Let $T_{\alpha}$ denote a minimum weight $\alpha$-colored $s$-arborescence in $G_s$. Also, let $T_{min}$ denote a minimum weight $s$-arborescence (without any constraint on colors) of $G_s$. It is easy to see that for each vertex $v \in V(G_s) \setminus \{s\}$, the edge entering $v$ in $T_{min}$ has the minimum weight among all edges entering $v$ in $G_s$ ($q = 1$ case of Lemma~\ref{lemma:minWt_predecessor}).

Notice that the conditions in Line~\ref{step:ReduceR} and Line~\ref{step:ReduceB} will not be simultaneously true since $\alpha_1 + \alpha_2 \geq n-1 = |V_R \cup V'_R \cup V_B \cup V'_B|$. If both the conditions are false, then $S = \emptyset$. Then, for every vertex $v \in V(G_s) \setminus \{s\}$, the edge of $T$ entering $v$ is a minimum weight edge of $G_s$ entering $v$.  By Observation~\ref{obs:spt-G}, $T_{min}$ is an $s$-arborescence of $G_s$ and hence, for every vertex $v \in V(G) \setminus \{s\}$, $T_{min}$ contains at least one edge of $G_s$ entering $v$. Therefore $w(T_{min}) \geq w(T)$ and hence $T$ is also a minimum weight $s$-arborescence. 

Now suppose the condition $|V_R \cup V'_R| > \alpha_1$ in Line~\ref{step:ReduceR} is true. In this case, since no vertex in $V_R$ has an incoming blue edge in $G_s$, at least $k = |V_R \cup V'_R| - \alpha_1$ vertices in $V'_R$ should get an incoming blue edge in $T_{\alpha}$. Let $S_{\alpha}$ denote the set of vertices in $V'_R$ which get an incoming blue edge in $T_{\alpha}$. Hence $w(T_{\alpha}) \geq w(T_{min}) + \sum_{v \in S_{\alpha}}(b_v - r_v)$. Since the set $S$ selected by the algorithm  consists of first $k$ vertices of $V'_R$ in the order of non-decreasing $(b_v - r_v)$ value, $\sum_{v \in S}(b_v - r_v) \leq \sum_{v \in S_{\alpha}}(b_v - r_v)$. Hence $w(T) = w(T_{min}) + \sum_{v \in S}(b_v - r_v) \leq w(T_{\alpha})$. Hence $T$ is also a minimum weight $\alpha$-colored $s$-arborescence of $G$. The case when the condition $|V_B \cup V'_B| > \alpha_2$ in Line~\ref{step:ReduceB} is true can be analysed similarly.
\end{proof}

Observe that Theorem~\ref{thm:minWt-red-blue} implies that \textsc{Min CC-SPT} on $2$-colored weighted digraphs $G$ with $n$ vertices, $m$ edges and only positive cycles can be solved can be solved in $\mathcal{O}(T_{SSP}(n,m)+T_{Sort}(n))$ time.

\section{Color-Constrained Shortest Paths}
\label{sec:sp}
In this section, we deviate slightly from our problem of focus and consider the related problem of finding an $\alpha$-colored shortest path starting from a specific source $s$ and ending at a specific destination $t$. 

\parprobNoPar{\textsc{Color-Constrained Shortest Path} (\textsc{CC-SP})}{}{A $w$-weighted digraph $G$ with coloring function $\chi:E(G) \rightarrow [q]$, two vertices $s,t\in V(G)$ and a tuple $\alpha=(\alpha_1,\ldots,\alpha_q)\in \mathbb{N}^q$.}{Does $G$ have an $\alpha$-colored shortest $s$-$t$-path?}
\\

A vertex variant of \textsc{CC-SP} has been recently studied by Bentert et al.~\cite{BentertKN23}. In the variant they considered, the input is a $q$-vertex-colored edge-weighted digraph $G$ and the question is to decide if $G$ has a shortest $s$-$t$-path with equal number of vertices of every color. It was shown that this problem is \NP-complete even when all edges of $G$ have weight one. We may further assume that the shortest path length $\ell$ is such that $\ell+1$ is a multiple of $q$, otherwise the input instance is trivially a \textsf{No}-instance. With this assumption, we can easily show that this problem has a polynomial-time reduction to the \textsc{Vertex Color-Constrained Shortest Path} problem defined below. 

\parprobNoPar{\textsc{Vertex Color-Constrained Shortest Path} (\textsc{VCC-SP})}{}{A $w$-weighted digraph $G$ with coloring function $\chi:V(G) \rightarrow [q]$, two vertices $s,t\in V(G)$ and a tuple $\alpha=(\alpha_1,\ldots,\alpha_q)\in \mathbb{N}^q$.}{Does $G$ have an $\alpha$-colored shortest $s$-$t$-path?}

The reduction to \textsc{VCC-SP} is simply to output the input graph $G$ and the color-constraint vector $\alpha$ in which $\alpha_i=(\ell+1)/q$ for each $i \in [q]$.  By this reduction, it follows that \textsc{VCC-SP} is \NP-complete. We now show that \textsc{CC-SP} and \textsc{VCC-SP} are in fact equivalent with respect to polynomial-time reductions. 

\begin{mylemma}
    \label{vertex_version_to_edge_version}
    Given an instance of \textsc{VCC-SP}, there is a linear-time algorithm that produces an equivalent instance of \textsc{CC-SP}.
\end{mylemma}
\begin{proof}
Consider an instance $I=(G,\chi,w,s,t,\alpha)$ of \textsc{VCC-SP} where $\chi:V(G)\rightarrow [q]$ is a coloring on $V(G)$. Let $G'$ denote the graph obtained from $G$ by adding a new vertex $s'$ as an in-neighbour of $s$. Let $w':E(G')\rightarrow \mathbb{R}$ be the weight function defined as $w'((s',s))=0$ and for every other edge $e$, $w'(e)=w(e)$. Let $\chi':E(G')\rightarrow [q]$ denote the coloring defined as $\chi'((u,v))=\chi(v)$ for each $(u,v) \in E(G')$. Let $I'=(G',\chi',w',s',t,\alpha)$ be the instance of \textsc{CC-SP} corresponding to $I$. Clearly, this reduction runs in linear time. 

It can be seen that for every $s$-$t$-path $P=(s=v_1,\ldots ,v_k=t)$ in $G$ there is a corresponding $s'$-$t$-path $P'=(s',s=v_1,\dots, v_k=t)$ in $G'$ and this correspondence is bijective. By the definition of $\chi'$, the number of vertices of color $i$ in $P$ is equal to the number of edges of color $i$ in $P'$ for each color $i \in [q]$. Further, $w(P)=w'(P')$.  From this, it follows that $I$ is a \textsf{Yes}-instance of  \textsc{VCC-SP} if and only if $I'$ is a \textsf{Yes}-instance of \textsc{CC-SP}. 
\end{proof}
Since \textsc{CC-SP} can be easily seen to be in \NP, the above reduction shows that \textsc{CC-SP} is \NP-complete.  Next, we show a reduction from \textsc{CC-SP} to \textsc{VCC-SP}.
\begin{mylemma}
\label{edge_version_to_vertex_version}
Given an instance of \textsc{CC-SP}, there is a linear-time algorithm that produces an equivalent instance of \textsc{VCC-SP}.
\end{mylemma}
\begin{proof}
Consider an instance $I=(G,\chi,w,s,t,\alpha)$ of \textsc{CC-SP} where $G$ is a $q$-colored digraph. Let $L(G)$ denote the directed line graph of $G$. For an edge $e \in E(G)$, $v_{e}$ denotes the vertex in $L(G)$ corresponding to $e$.  By definition, there is an edge from a vertex $v_{e}$ to a vertex $v_{f}$ in $L(G)$ if and only if there is a vertex $x$ in $G$ and edges $e,f$ such that $e$ enters $x$ and $f$ leaves $x$.  

Let $X$ be the subset of vertices of $L(G)$ which correspond to the edges leaving $s$ in $G$ and $Y$ be the subset of vertices of $L(G)$ which correspond to the edges entering $t$ in $G$. 
Let $G'$ be the digraph obtained from $L(G)$ by adding two new vertices $s'$, $t'$ such that there is an edge from $s'$ to every vertex in $X$ and an edge from every vertex in $Y$ to $t'$.  This construction establishes a bijective mapping between $s$-$t$-paths in $G$ and $s'$-$t'$-paths in $G'$, described as follows. Consider any $s$-$t$-path $P=(s=v_1,v_2,\ldots, v_k=t)$ in $G$. Let $E(P)=\{e_1,\ldots,e_{k-1}\}$ where $e_i=(v_i,v_{i+1})$ for $i \in [k-1]$. The path $P'$ in $G'$ corresponding to $P$ is defined as the path $(s',v_{e_1},\ldots,v_{e_{k-1}},t')$. It can be verified that this is indeed a bijective correspondence. 

Let $w'$ be the weight function $E(G')\rightarrow \mathbb{R}$ defined as follows. For each edge $(s',v_e) \in E(G')$ leaving $s'$, set $w'((s',v_e))=w(e)$ and for each edge $(v_e,t')\in E(G')$ entering $t'$, set $w'((v_e,t'))=w(e)$. For any other edge $(v_e,v_f)\in E(G')$ where $e,f\in E(G)$, set $w'((v_e,v_f))=w(e)+w(f)$. 
Observe that the weight function $w'$ is defined in such a way that for any $s$-$t$-path $P$ in $G$, its corresponding path $P'$ in $G'$ satisfies $w'(P')=2w(P)$.  

Define a coloring $\chi':V(G') \rightarrow [q]$ from $\chi$ by setting $\chi'(s')=\chi'(t')=1$ and $\chi'(v_e)=\chi(e)$ for each $e\in E(G)$. Let $\alpha'$ be $(\alpha_1+2, \alpha_2, \ldots,\alpha_q)$. Note that for any $s$-$t$-path $P$ in $G$ and the corresponding $s'$-$t'$-path $P'$ in $G'$, $P$ has exactly $\alpha_i$ edges of color $i$ for each $i \in [q]$ if and only if $P'$ has exactly $\alpha_i$ vertices of color $i$ for each $i \in [q] \setminus \{1\}$ and exactly $\alpha_1+2$ vertices of color 1.  

Let $I'=(G',\chi',w',s',t',\alpha')$ be the instance of \textsc{VCC-SP} corresponding to $I$. Clearly, this reduction runs in linear time. From the definition of the bijective mapping of paths discussed above, it also follows that every $s$-$t$-path $P$ in $G$ bijectively corresponds to an $s'$-$t'$-path $P'$ in $G'$ such that $w'(P')=2w(P)$ and $P$ has at most $\alpha_i$ edges of color $i$ for each $i \in [q]$ if and only if $P'$ has at most $\alpha'_i$ vertices of color $i$ for each $i \in [q]$.  Hence, $I$ is a \textsf{Yes} instance of  \textsc{CC-SP} if and only if $I'$ is a \textsf{Yes} instance of \textsc{VCC-SP}. 
\end{proof}

Bentert et al.~\cite{BentertKN23} gave an $\mathcal{O}(m n^q)$ time algorithm for a variant of \textsc{VCC-SP}, where the edge weights considered were natural numbers and both upper bound and lower bound constraints are given on the number of vertices of each color in a shortest path produced as output. The correctness arguments given by Bentert et al.~\cite{BentertKN23} do not depend on the weights being natural numbers and hold even when the input graph real-weighted with no negative-weight cycles. Hence, for such graphs, we can use the algorithm by Bentert et al.~\cite{BentertKN23} with lower bounds for each color taken as zero, to solve \textsc{VCC-SP}. Observe that for an input graph $G$ with $n$ vertices and $m$ edges, in the graph $G'$ obtained by the reduction given in the proof of Lemma~\ref{edge_version_to_vertex_version}, the number of vertices is $\mathcal{O}(m)$ and the number of edges is $O(mn)$. Therefore, by Lemma \ref{vertex_version_to_edge_version} and Lemma \ref{edge_version_to_vertex_version}, 
we have the following result. 

\begin{mytheorem}
\label{thm:spt-hard-easy}
    \textsc{CC-SP} is \NP-complete and can be solved in $\mathcal{O}(nm^{q+1})$ time when the input graph has no negative-weight cycles. 
\end{mytheorem}

From a parameterized complexity perspective, we observe that the reduction from \textsc{VCC-SP} to \textsc{CC-SP} described in Lemma \ref{vertex_version_to_edge_version} is also a parameter-preserving reduction where the parameter is the number of colors of the input graph. Therefore, the \W[1]-hardness of \textsc{VCC-SP} with respect to the number of colors of the input graph as the parameter~\cite{BentertKN23} also holds for \textsc{CC-SP}.

\section{Concluding Remarks}
\label{sec:concl}
The complexity of \textsc{CC-ARB} in general digraphs for fixed $q$ is an interesting direction of research. Recall that \textsc{CC-ARB} is polynomial-time solvable when $q=2$ and \NP-hard for arbitrary $q$. From a parameterized complexity perspective, $q$ is the most natural parameter for \textsc{CC-ARB} and the fixed-parameter tractability of the problem is open. On a related line, determining the existence of a color-constrained shortest $s$-$t$-path turns out to be fixed-parameter tractable with the length of the path as the parameter but fixed-parameter intractable with $q$ as the parameter. The intractability follows since the above mentioned reduction from the vertex-colored variant is parameter-preserving.

The most intriguing color-constrained subgraph, at least to us, is an $\alpha$-colored perfect matching in a bipartite graph. The problem of finding a $(k, n/2 - k)$-colored perfect matching in an $n$-vertex graph, often called \textsc{Exact Perfect Matching}\footnote{A related problem with edges assigned to (possibly overlapping) sets instead of color classes, was already studied in 1977 by Itai, Rodeh and Tanimoto~\cite{itai1977some, tanimoto1978some} under the name \textsc{Restricted Perfect Matching}. They showed that for arbitrary number of sets, this problem is \NP-complete. However, their reduction produces instances where the sets are overlapping. They also showed that in the special case when $q = 2$ and only one set has a constraint, that is when $\alpha = (\alpha_1, n/2)$, the problem admits a polynomial-time solution. They conjectured that the $q = 2$ case when both sets are constrained is \NP-complete.}, was studied by Papadimitriou and Yannakakis \cite{PapadimitriouY82} in 1982. Mulmuley, Vazirani and Vazirani \cite{mulmuley1987matching} gave a randomized polynomial-time algorithm for \textsc{Exact Perfect Matching} in 1987. This problem still remains as one of those rare examples of a problem in \RP\ not yet known to be in \P.

\bibliography{ref}

\appendix
\section{Maximum Flow with Lower and Upper Bound Constraints}

In this section, we give a proof of Lemma~\ref{lem:feasible-flow-max-flow} in detail. Note that the construction used in the proof is a suitable adaptation of a well-known construction that reduces the problem of finding a  feasible flow to a maximum flow computation (see for example~\cite{Erickson15}). 
\feasibleflow*
\begin{proof}

    The network $H_k$ is constructed from $H$ by adding a new source vertex $s'$ and a new sink vertex $t'$ such that there is an edge from $s'$ to every vertex in $V(H)$ and an edge from every vertex in $V(H)$ to $t'$. 
     The capacity function $c' : V(H_k) \times V(H_k) \rightarrow \mathbb{R}^{\geq 0}$ is defined as follows.
$$
c'(u,v)= 
\begin{cases}
\underset{x\in N^{-}(v)}\sum d(x,v) & \text{if } u=s' \text{ and } v\in V(H)\setminus \{s\} \\
\underset{y\in N^{+}(u)}\sum d(u,y) & \text{if } v=t' \text{ and } u\in V(H)\setminus \{t\} \\
c(u,v) - d(u,v) & \text{if }(u, v) \in E(H) \\
k & \text{if }(u,v)= (s', s) \text{ or } (u, v) = (t, t')
\end{cases}
$$

Suppose $f: E(H) \rightarrow \mathbb{R}^{\geq 0}$ is a feasible flow in $H$ with $|f|=k$. Consider the following function $f': E(H_k) \rightarrow \mathbb{R}^{\geq 0}$:

$$
f'(u,v)= 
\begin{cases}
\underset{x\in N^{-}(v)}\sum d(x,v) & \text{if } u=s' \text{ and } v\in V(H)\setminus \{s\} \\
\underset{y\in N^{+}(u)}\sum d(u,y) &\text{if } v=t' \text{ and } u\in V(H)\setminus \{t\} \\
f(u,v)- d(u,v) & \text{if }(u, v) \in E(H) \\
k & \text{if }(u,v)= (s', s) \text{ or } (u, v) = (t, t')
\end{cases}
$$

\noindent Now we will show that $f'$ is a  flow in $H_k$. \\

\noindent\textit{Capacity and Non-negativity Constraints: } 
By definition, the capacity and non-negativity constraints are satisfied for the edges in $E(H_k)\setminus E(H)$. Since $f$ is a feasible flow in $H$, for any edge $(u,v)\in E(H_k)\cap E(H)$,  $d(u,v)\leq f(u,v)\leq c(u,v)$. Therefore, $0\leq f(u,v)-d(u,v)\leq c(u,v)-d(u,v)$ i.e., $0\leq f'(u,v)\leq c'(u,v)$. Thus, the capacity and non-negativity constraints on $E(H_k)$ are satisfied.\\

\noindent\textit{Flow Conservation: }
For the vertex $s\in V(H_k)$,
\begin{quote}
$\underset{u\in N_{H_k}^-(s)}\sum f'(u,s)= f'(s',s)= k$ and 
\end{quote}
\begin{align*}
\underset{v\in N_{H_k}^+(s)}\sum f'(s,v) &= f'(s,t')+\underset{\substack{v\in N_{H_k}^+(s)\setminus \{t'\}}}\sum f'(s,v)\\
&= \underset{v \in N_{H}^+(s)}\sum d(s,v)+ \underset{v\in N_{H}^+(s)}\sum (f(s,v)-d(s,v))\\ &= \underset{v\in N_{H}^+(s)}\sum f(s,v) \\ &=k 
\end{align*}
Thus, the conservation constraint is satisfied for the vertex $s$. We can similarly show that the conservation constraint is also satisfied for the vertex $t\in V(H_k)$. Now, consider any vertex $v\in V(H_k) \setminus \{s',t',s,t\}$. We can obtain the flow conservation at $v$ as follows. 

\begin{align*}
    \underset{u\in N_{H_k}^-(v)}\sum f'(u,v)&= f'(s',v)+ \underset{\substack{u\in N_{H_k}^-(v)}\setminus \{s'\}}\sum f'(u,v)\\
    &= \underset{u\in N_{H}^-(v)}\sum d(u,v) + \underset{u\in N_{H}^-(v)}\sum (f(u,v)-d(u,v))\\
    &=\underset{u\in N_{H}^-(v)}\sum f(u,v)
\end{align*}

and 
\begin{align*}
    \underset{w\in N_{H_k}^+(v)}\sum f'(v,w)&= f'(v,t')+ \underset{\substack{w\in N_{H_k}^+(v)} \setminus \{t'\}}\sum f'(v,w),\\
    &= \underset{w\in N_{H}^+(v)}\sum d(v,w) + \underset{w\in N_{H}^+(v)}\sum (f(v,w)-d(v,w)),\\
    &=\underset{w\in N_{H}^+(v)}\sum f(v,w).
\end{align*}

\noindent Since $f$ is a valid flow in $H$, by flow conservation property of $f$, $\underset{u\in N_{H}^-(v)}\sum f(u,v)=\underset{w\in N_{H}^+(v)}\sum f(v,w)$ for each $v\in V(H)$. Hence, it follows that the flow conservation constraints are satisfied by $f'$ for all the vertices in $V(H_k)\setminus \{s', t'\}$.

Now, we will look at the value of $f'$.  
\begin{eqnarray*}
& |f'| = \underset{v\in V(H_k)} \sum f'(s',v) &= f(s', s) + \underset{v \in V(H)\setminus \{s\}} \sum f'(s', v)\\
& &= k + \underset{v \in V(H)\setminus \{s\}}  \sum ~ \underset {x\in N^{-}_H(v)} \sum d(x, v) \\
& &= k + \underset{(u, v)\in E(H)} \sum d(u, v) \hspace{0.2cm} 
\end{eqnarray*}
The last equality uses our assumption that there are no incoming edges to $s$ in $H$. This completes the proof that if $H$ has a feasible flow of value $k$, then $H_k$ has a flow $f'$ of value $k+ \underset{(u, v)\in E(H)} \sum d(u, v)$.

Conversely, suppose $H_k$ has a $(s',t')$-flow $f'$ of value $k+ \underset{(u, v)\in E(H)} \sum d(u, v)$.  
We define a function $f: E(G)\rightarrow \mathbb{R}^{\geq 0}$ on $H$ as $$f(u,v)= f'(u,v)+ d(u,v), \text{ for all }(u,v)\in E(H).$$
For every edge $(u,v)\in E(H)$, the function $f$ satisfies the lower bound and upper bound constraints because we have $0\leq f'(u,v)\leq c'(u,v)=c(u,v)-d(u,v)$ and it follows that $d(u,v)\leq f(u,v)\leq c(u,v)$.

It remains to show that $f$ satisfies flow conservation constraint at each vertex $v \in V(H) \setminus \{s,t\}$.  For this, we make some simple observations. 

\begin{eqnarray*}
    & \underset{v\in V(H_k)} \sum f'(s',v) = |f'| &= k+ \underset{(u, v)\in E(H)} \sum d(u, v) \text{ (By assumption)}\\
    & & = k + \underset{v \in V(H)\setminus \{s\}}  \sum ~ \underset {x\in N^{-}_H(v)} \sum d(x, v) \\
    & & = \underset{v \in V(H_k)} \sum c'(s', v)
\end{eqnarray*}
Hence, on every edge of the form $(s', v)$, $f'(s',v)=c(s',v)$. 
This means, 
\begin{equation}\label{eq:s-s'-flow}
    f'(s', s)= c'(s', s) =   k 
\end{equation} and
\begin{equation}\label{eq:inflow}
\text{for any } v \in V(H) \setminus {s}, f'(s', v)= c'(s', v) =   \underset{x\in N^{-}_H(v)}\sum d(x,v)
\end{equation}
Similarly, we can also show that  on every edge of the form $(u, t')$, $f'(u,t')=c(u,t')$. Therefore,
\begin{equation}\label{eq:outflow}
\text{for any } u \in V(H) \setminus {t}, f'(u, t')= c'(u, t') =   \underset{y\in N^{+}_H(u)}\sum d(u,y)
\end{equation}

\noindent Now, consider any vertex $v \in V(H) \setminus \{s,t\}$.
\begin{eqnarray*}
  &\underset{u\in N_{H}^-(v)}\sum f(u,v)&= \underset{u\in N_{H}^-(v)}\sum (f'(u,v)+d(u,v))\\ 
  & &=\underset{u\in N_{H}^-(v)}\sum f'(u,v)+ \underset{u\in N_{H}^-(v)}\sum d(u, v)\\
  & &=\left(\underset{u\in N_{H}^-(v)}\sum f'(u,v)\right)+f'(s',v) \text{\hspace{1cm} (By Eq. ~\ref{eq:inflow})}\\ 
  & &=\underset{u\in N_{H_k}^-(v)}\sum f'(u,v) 
\end{eqnarray*}
Similarly, for any vertex $v \in V(H) \setminus \{s,t\}$, using Eq.~\ref{eq:outflow}, we can show that $$\underset{w\in N_{H}^+(v)}\sum f(v,w)=\underset{w\in N_{H_k}^+(v)}\sum f'(v,w).$$
 By the flow conservation property of $f'$, it concludes that $f$ satisfies the conservation constraint for each vertex of $v\in V(H) \setminus \{s,t\}$. \\\\
Thus, $f$ is a feasible flow in $H$. It remains to show that the value of the flow $f$ is $k$. 

\begin{eqnarray*}
    &|f|=\underset{v\in N_{H}^+(s)}\sum f(s,v)&= \underset{v\in N_{H}^+(s)}\sum f'(s,v)+ \underset{v\in N_{H}^+(s)}\sum d(s,v)\\
    & &= \underset{v\in N_{H}^+(s)}\sum f'(s,v)+f'(s,t')\\
    & &=\underset{v\in N_{H_k}^+(s)}\sum f'(s,v)\\
    & &= f'(s',s) \hspace{1cm} \text{(By the flow conservation property of $f'$ at $s$)}\\
    & & = k \hspace{1cm} \text{(By Eq.~\ref{eq:s-s'-flow})}
\end{eqnarray*}
Thus, we have shown that if $H_k$ has a flow of value $k+ \underset{(u, v)\in E(H)} \sum d(u, v)$, then $H$ has a feasible flow of value $k$.
\end{proof}

\begin{remark}\label{rmk:generalization}
        It is possible to generalize the above lemma to reduce the problem of finding a minimum cost feasible flow of value $k$ to the problem of finding a minimum cost maximum flow.
        Consider the flow network $(H, s, t, c, d, \widehat{c})$ where $\widehat{c}$ is the cost function. To find a minimum cost feasible flow of value $k$ in $H$, we construct $H_k$ as in the proof of Lemma~\ref{lem:feasible-flow-max-flow} with the cost function $\tilde{c}$ defined as follows. For each edge $(u, v) \in E(H)$, $\tilde{c}(u,v) = \widehat{c}(u,v)$ and for the remaining edges of $H_k$, the costs are zero. With this modification, from a flow $f$ in $H$, we will be able to get a flow $f'$ in $H_k$ in the same way as in the previous proof. Since the cost of all edges in $E(H_k) \setminus E(H)$ are zero, the cost of $f'$ now becomes $\underset{(u, v)\in E(H_k)}\sum \tilde{c}(u,v) f'(u, v)= \underset{(u, v)\in E(H)}\sum \widehat{c}(u,v) (f(u,v )- d(u, v))= \underset{(u, v)\in E(H)}\sum \widehat{c}(u,v) f(u,v )-  \underset{(u, v)\in E(H)} \sum \widehat{c}(u,v) d(u,v)$. Since the second term in the summation is independent of $f$, in order to find a minimum cost feasible flow of value $k$ in $H$, it suffices to find a minimum cost maximum flow in $H_k$.  
\end{remark}

\end{document}